%% file: main.tex
\documentclass[runningheads]{llncs}

\input{preamble/extrapackages}

\input{preamble/macrosetup}
\begin{document}

\title{A Formal Framework for Predicting Distributed System Performance under Faults\\
(Extended Version)
}
\titlerunning{Predicting Distributed System Performance under Faults}

\author
{Ziwei Zhou\inst{1}\thanks{Joint first authors with equal contribution.}
\and
Si Liu\inst{2}\repeatthanks
\and
Zhou Zhou\inst{1}
\and 
Peixin Wang\inst{1}
\and 
Min Zhang\inst{1}
}

\authorrunning{Z. Zhou et al.}

\institute{
East China Normal University, China
\and
%ETH Zurich, Switzerland
Texas A\&M University, USA
}

\maketitle              

\input{sections/abs}
\input{sections/intro}
\input{sections/prelim}

\input{sections/faults}

\input{sections/impl}

\input{sections/case-study}

\input{sections/related}
\input{sections/concl}

\section*{Acknowledgments}
We thank the anonymous reviewers for their valuable feedback. 
Min Zhang was funded by the NSFC Projects (No. 62372176 and No. 92582108) and Shanghai Trusted Industry Internet Software Collaborative Innovation Center. 
Peixin Wang was sponsored by CCF-Huawei Populus Grove Fund.
Si Liu and Min Zhang 
are the corresponding authors.

%\newpage
\bibliographystyle{splncs04}
\bibliography{ref}

\newpage
\appendix
\input{sections/appendix}

% \input{sections/f-trans}

\end{document}

%% file: preamble/extrapackages.tex
\usepackage{graphicx}
\usepackage{colortbl}
\usepackage{tabularx}
\usepackage{alltt}
\usepackage{soul}
\usepackage{wrapfig}
\usepackage{subcaption}
\usepackage{framed}
\usepackage{quoting}
\usepackage[ASCII]{maude}
\usepackage{xspace}
\usepackage{enumitem}
\usepackage{booktabs}
\usepackage[T1]{fontenc}
\usepackage{graphicx}
\usepackage{pgfplots}
\pgfplotsset{compat=1.18} 
\usepackage{filecontents}
\usepackage{url}
\usepackage{amsmath,amssymb}
\usepackage{float}

\usepackage{listings}
\lstdefinelanguage{Maude}{
    morekeywords={mod, is, endm, then, else, class, including, sort, sorts, subsort, subsorts, var, vars, op, ops, eq, ceq, rl, crl, owise, assoc, comm, id, ctor, if, fi, and, or, inc, nonexec},
    sensitive=true, % keywords are not case-sensitive
    morecomment=[l]{***} % l is for line comment
}

\lstdefinestyle{maudestyle}{
    commentstyle={\it\color{codegreen}},
    keywordstyle={\color{blue}},
    numberstyle=\tiny\color{gray},
    basicstyle=\ttfamily\scriptsize,
    breakatwhitespace=false,         
    breaklines=true,                 
    captionpos=b,                    
    keepspaces=true,                 
    numbers=left, % none
    numbersep=5pt,                  
    showspaces=false,                
    showstringspaces=false,
    showtabs=false,                  
    tabsize=2,
    moredelim=**[is][\color{red}]{@}{@}
}

%% file: preamble/macrosetup.tex
\newcommand{\zzw}[1]{{\color{orange}{[Zzw: #1]}}}
\newcommand{\zhang}[1]{{\color{blue}{[Min: #1]}}}
\newcommand{\todo}[1]{{\color{red}{[#1]}}}
\newcommand{\tocite}[1]{{\color{red}{[CITE]}}}

\newcommand{\inlsec}[1]{\smallskip\noindent\textbf{#1.}}
\newcommand{\inlsecit}[1]{\smallskip\noindent\textit{#1.}}
\newcommand{\ourtool}{\textsc{PerF}\xspace}

\newcommand{\circled}[1]{{\textcircled{\scriptsize #1}}}

\definecolor{shadecolor}{gray}{0.95}

\newcommand{\repeatthanks}{\textsuperscript{\thefootnote}}

\definecolor{codegreen}{HTML}{2E8B57}
\definecolor{mycolor3}{HTML}{D18CB8}
\definecolor{mycolor15}{HTML}{DADBEB}

%% file: sections/abs.tex
\begin{abstract}
Today's distributed systems  operate in complex environments that inevitably involve faults and even adversarial behaviors.
Predicting their performance under such  environments  directly from formal designs remains a long-standing challenge. 
We present the first formal framework that systematically enables performance prediction of distributed systems across diverse faulty scenarios. 
Our framework features a fault injector together with a wide range of faults, reusable 
as a library, and 
model compositions that integrate the system and the fault injector 
into a unified 
model suitable for statistical analysis of  performance properties such as throughput and latency. 
We formalize the framework in Maude and implement it as an automated tool, \ourtool.
Applied 
to representative distributed systems,  \ourtool accurately predicts system performance under varying fault settings, with estimations from formal designs consistent with  evaluations on real deployments.

\end{abstract}

%% file: sections/intro.tex
\section{Introduction} \label{sec:intro}

Distributed systems, such as cloud databases and blockchains, form the backbone of today's digital infrastructure. 
Formal methods have proven highly effective in verifying the correctness of their designs~\cite{DBLP:books/aw/Lamport2002,DBLP:conf/fase/Meseguer25}.
%Especially at an early design stage, formal analysis can not only uncover subtle flaws that are otherwise hard to detect,but also enable systematic exploration of the design space~\cite{DBLP:books/aw/Lamport2002,DBLP:conf/fase/Meseguer25}. 
\emph{Performance, however, is equally critical.}
The need to predict a system's performance (e.g., throughput and latency)
in realistic distributed environments, ideally before its implementation and deployment, has long been recognized by both academia~\cite{alur-henzinger-vardi,bobba2018survivability} and industry~\cite{DBLP:series/isc/Chuat22,Davis2025TLAStat,Newcombe2015}. 
Such environments inevitably involve faults, including network partitions, message delays, and even Byzantine behaviors~\cite{DBLP:books/mk/Lynch96,DBLP:books/daglib/0019513}.
In particular, Amazon Web Services has been seeking, 
for over a decade,  a feasible way to model distributed systems and predict their performance degradation~\cite{Brooker2024TLAConf,Newcombe2015}. 

\inlsec{State-of-the-Art}
Despite  attempts to address this long-standing challenge, it remains largely unresolved.
First, most quantitative formal analyses of distributed systems assume a fault-free environment~\cite{bobba2018survivability,DBLP:conf/tase/LiangL21,liu-tosem,DBLP:journals/lites/LiuGRNGM17,DBLP:journals/pacmpl/LiuMOZB22,icfem17,DBLP:journals/fac/LiuOWGM19,DBLP:conf/wrla/LiuOWM18,qmaude,Vanlightly2022TLCSim}.
Yet
ignoring faults can introduce a significant gap between model-based predictions and observed system performance.
For instance,  two distributed  transaction protocols that appear competitive under a recent fault-free analysis~\cite{DBLP:journals/pacmpl/LiuMOZB22} diverge sharply once message loss is introduced (see Section~\ref{sec:exp}). 
Second, approaches that incorporate faults  are often \textit{ad hoc}~\cite{DBLP:conf/ccs/AlTurkiKKNST18,DBLP:conf/fm/AlturkiR19,DBLP:conf/fase/EckhardtMAMW12,liu-sigcomm,n-tube}, manually intertwining fault behaviors with system semantics.
This hinders reuse,
e.g., 
applying the same fault to a different protocol requires rebuilding the entire model from scratch.  
Third, all of these approaches consider only a few isolated fault types, whereas  systems in practice often experience a wide range of faults that may also coexist, e.g., network partitions coinciding with malicious attacks.
Therefore, existing analyses capture only a narrow fragment of the overall system performance space.

All of these limitations call for a systematic formal framework that
(i) provides \emph{comprehensive fault coverage} across diverse %benign (e.g., message loss and network partitions) and Byzantine (e.g., tampering and equivocation)
faults encountered in practice;
(ii) ensures \emph{modularity} by
decoupling  individual fault  modeling from system modeling 
to support
composition, reuse, and extensibility;
and 
%(iii) supports \emph{extensibility} through well-defined interfaces for integrating new fault types; and
(iii) enables \emph{automatic} generation of 
fault-injected system models
and quantitative analysis of performance properties, e.g.,  using statistical verifiers.

\inlsec{Our Solution}
We present the \emph{first} formal framework  that  realizes these goals.
At its core lies \emph{a library of faults} covering a wide range, including both benign faults (e.g., message loss and network partitions) and Byzantine faults (e.g., tampering and equivocation).
We model each fault as a message-passing \emph{actor system}~\cite{actors}, a design choice that is motivated  by 
 %both an observation and an 
our two insights.
First, virtually all distributed systems can be  specified as actor systems, where nodes, such as clients and servers, 
communicate via asynchronous message exchanges.
%~\cite{DBLP:conf/fase/Meseguer25,DBLP:journals/pacmpl/LiuMOZB22,DBLP:conf/pldi/DesaiGJQRZ13,DBLP:conf/ispdc/Agha14}.
Modeling faults as actors allows their effects on a system to be naturally captured by message-triggered interactions between the fault injector and system actors. 
Second, most practically relevant faults in distributed systems 
%share a common essence: they 
can  be commonly viewed as manipulations of messages between actors.
For instance, 
a node crash corresponds to dropping all of the messages sent to a specific actor;
a network partition can be seen as discarding all of the messages between two %disjoint 
actor sets, while preserving communication within each set.
This unified view not only simplifies the modeling of individual faults
but also allows different fault behaviors to be seamlessly integrated into the  framework for  performance analysis.

How do faults, then, interact with the  distributed system under analysis?
We realize this interaction through  \emph{model composition}.
Specifically, our framework takes as input a  system model and
a fault injector incorporating 
one or more faults, and constructs an integrated model that preserves the original system semantics while enriching it with relevant fault  behaviors.
To account for realistic environments in which multiple faults may arise simultaneously (e.g., a network partition alongside equivocation), we introduce \emph{fault-behavior priority levels} for fine-grained control over their interactions.
This avoids semantic inconsistencies such as delivering an equivocated message across partitions. 
Moreover, we design the model composition to also preserve the \emph{absence of nondeterminism} property~\cite{pmaude}, %\footnote{Intuitively, absence of nondeterminism means, for any given system state, the subsequent transition is uniquely determined, without any nondeterministic branching.}
thereby supporting \emph{statistical model checking} analysis, a formal approach that offers stronger guarantees  for quantitative results than pure simulation and scales  to large distributed systems~\cite{agha18,Sen05}. 
To facilitate exploring performance space,
 our framework also allows
 % provides
%(i) configurable fault injection, allowing
users to configure which faults to inject and when to inject them and
%(ii) 
provides a monitor 
%(extending the  mechanism introduced in~\cite{DBLP:journals/pacmpl/LiuMOZB22})
that records events of interest at  runtime, based on which users can readily define performance properties. 
All of these features are integrated into our automated tool \ourtool.

Our approach is formalism-independent. 
In this work, we instantiate it in the Maude specification language and  tool~\cite{maude-book} that has been successfully applied to a wide range of
 distributed systems~\cite{DBLP:journals/pacmpl/LiuMOZB22,maude2025applications}.
 Specifically, both the  distributed system under analysis and the faults are formalized as \emph{probabilistic rewrite theories}~\cite{pmaude}.
Since performance properties, such as throughput and latency, intrinsically involve time, probabilities are primarily used to model message delays, which are sampled from certain probability distributions.
The model compositions are then realized as compositions of rewrite theories.

\inlsec{Contributions} 
Overall, this work makes the following contributions:

\begin{itemize}
    \item At the conceptual level, we address the long-standing challenge of predicting the performance of  distributed systems under realistic faulty environments
    directly from their formal designs.

\item At the technical level, we develop an actor-based formal framework in Maude, featuring a reusable  library of practically relevant faults. %\todo{(Section~\ref{sec:})}.
We   design a model composition that integrates the system and the fault injector managing these fault behaviors
into a unified %, correct-by-construction
model amenable to quantitative analysis.

\item At the practical level, we implement our framework as an automated tool,  called \ourtool,  that supports configurable fault injection and end-to-end  performance prediction through statistical model checking.
    
\end{itemize}

We %showcase %a series of 
%6 case studies  applying 
apply \ourtool to six representative distributed systems of different kinds.
Experimental results demonstrate
%In particular, we show 
that it can accurately predict system performance under various types and combinations of faults, with model-based estimations  aligning closely with empirical evaluations on real deployments
(Section~\ref{sec:exp}).

%\begin{comment}
This work also complements  long-established, \textit{post hoc} fault-injection  efforts on  distributed system deployments 
%~\cite{jepsen,twins,ChaosMonkey,fis,cofi,CrashFuzz,chronos,Rainmaker,10.1145/3576915.3623071,Ozkan-bft},
%which are, however, posterior and focused on correctness issues.
%thereby remaining reactive to potential performance issues.
%A detailed discussion is provided in Section~\ref{sec:discuss}.
by offering \emph{proactive} insight into system dynamics under faults, already at the design stage. See Section~\ref{sec:discuss} for a discussion.
%\end{comment}

%% file: sections/prelim.tex
\section{Preliminaries}
\label{sec:prelim}
%In this section, we present our running example, a simple voting protocol specified in Maude.
%We begin by recalling the necessary preliminaries.
%which we use to introduce the necessary preliminaries \todo{actor paradigm, probabilistic rewrite theories, Maude} (and to illustrate our framework in later sections).

\textbf{Maude.} As an executable formal specification language and analysis tool, Maude \cite{maude-book} has been successfully applied to a broad spectrum  of distributed systems.
A Maude module specifies a 
 \emph{rewrite theory}  \cite{DBLP:journals/tcs/Meseguer92}
 $\mathcal{R}=(\Sigma,E,L,R)$, where:
 
\begin{itemize}
\item $\Sigma$ is an algebraic \emph{signature}, i.e., a set of 
\emph{sorts}, \emph{subsorts}, and \emph{function symbols}; 
\item $(\Sigma, E)$ is an \emph{order-sorted equational
logic theory} \cite{DBLP:journals/tcs/GoguenM92}
 specifying the system's data types,  where  $E$ is a set of
(possibly conditional) equations and axioms; %,  and $B$ a set of
                                %equational axioms such as
                                %associativity, commutativity, and
                                %identity, so that equational
                                %deduction is performed \emph{modulo}
                                %the axioms $B$.  
\item $L$ is a set of rule \emph{labels}; and
\item $R$ is a collection of {\em labeled conditional rewrite rules\/}
  \( [l]:\,t\longrightarrow t'\,\mbox{ \textcolor{blue}{if} } \mathit{cond}\), with
  $t, t'$ $\Sigma$-terms and $l \in L$. These rules specify the system's local transitions.
\end{itemize}
We summarize the fragment of Maude syntax used in this paper and refer the reader to~\cite{maude-book} for details. 
Operators are declared  by
\texttt{\textcolor{blue}{op}} $f$ \texttt{:} $s_1 \ldots s_n$ \texttt{->}
$s$ and may have user-defined syntax, 
where `\verb+_+'  denotes argument positions, as in \verb@_+_@.
(Unconditional and conditional) equations and rewrite rules are introduced with the keywords
\textbf{\texttt{\textcolor{blue}{eq}}} and \texttt{\textcolor{blue}{ceq}}, and
\textbf{\texttt{\textcolor{blue}{rl}}} and \texttt{\textcolor{blue}{crl}}, respectively.
%for conditional rules.  
% An equation $f(t_1, \ldots, t_n) = t$ with the \texttt{owise} (for
% ``otherwise'') attribute  can be applied to a term $f(\ldots)$ 
% only if no other equation with lefthand side $f(u_1, \ldots, u_n)$ 
% can be applied. 
% Equations and rules  with the \texttt{nonexec} attribute 
% are ignored by the Maude rewrite engine, which  can for example be used
% at the metalevel for controlled execution.
%Variables %in such statements
%are declared with the keywords \texttt{\textcolor{blue}{var}} and
%\texttt{\textcolor{blue}{vars}}.  % or can be declared on-the-fly, having the
                          % form
                          % \texttt{\(\mathit{var}\):\(\mathit{sort}\)}. 
Modules can be imported as submodules,
e.g., by using the keyword \texttt{\textcolor{blue}{inc}}.
%  A  module in Maude can
% be imported as a submodule of another % in three different modes,
% protecting, extending, or including, expressed by the keywords
% \texttt{pr}, \texttt{ex}, or \texttt{inc}, respectively.  
% using keywords such as  \texttt{inc}.
Comments %in Maude
start with  `\texttt{\textcolor{codegreen}{***}}'.

\inlsec{Actors}
We follow Agha's message-passing \emph{actor} paradigm~\cite{actors} to specify systems. 
In Maude,
a  declaration \texttt{\textcolor{blue}{\,\,class} \(C\) |
  \(\mathit{att}_1\)\,\,:\,\,\(\mathit{s}_1\), \dots ,
  \(\mathit{att}_n\)\,\,:\,\,\(\mathit{s}_n\)\,\,} 
declares an actor (or object) class $C$  with attributes $att_1$ to $att_n$ of
sorts $s_1$ to $s_n$.  
An instance of class $C$  is  a term $\,\,\texttt{<}\, o : C\,\texttt{|}\,\mathit{att}_1: \mathit{val}_1, \dots , \mathit{att}_n: \mathit{val}_n\,\texttt{>}\,\,$,
where $o$ (of sort \texttt{Oid})  is the object's identifier,
and  $val_1$ to  $val_n$ are the values of 
the attributes.
Messages exchanged between objects are terms of sort \texttt{Msg}. 
A system state is modeled as a term of
 sort \texttt{Config} and 
is structured as a \emph{multiset} of objects and messages,
formed via the (juxtaposition) multiset union operator
 \texttt{\_\_}.
 
The dynamic behavior of a system is axiomatized by specifying each of its transition patterns through a rewrite rule.
For instance, the  rule  labeled \texttt{[reply]}

\begin{lstlisting}[language=maude,numbers=none]
rl [reply] :  (read(KEY) from O' to O)    < O : Server | database: DB >
          =>  < O : Server | >    (<KEY, DB[KEY]> from O to O') .
\end{lstlisting}

\noindent  defines a family of transitions in which a message \texttt{read(KEY) from O' to O} sent by the  client object \texttt{O'} is consumed
by the  server object \texttt{O}.
The server then looks up its   database and replies with the corresponding key-value pair via the message \texttt{<KEY, DB[KEY]> from O to O'}.
Note that 
attributes whose values do not change, such as \texttt{database}, can be omitted in the right-hand side of a
rule.
We provide 
a Maude specification of the Two-Phase Commit  protocol 
 in
 Appendix~\ref{appendix:tool}, 
 %\cite[Appendix~\ref{appendix:tool}]{tech-rpt},
 which we use to illustrate our framework and tool.

\inlsec{Probabilistic Rewrite Theories}
Extending the original rewrite theories, \emph{probabilistic rewrite theories}~\cite{pmaude} can specify a broader 
class of systems with probabilistic behavior.
%including object-based probabilistic real-time systems. %, which are the focus of this work.
In our framework, probabilities arise from message  delays sampled from probability distributions \emph{and} from probabilistic choices during fault injection, e.g., the likelihood of dropping a message.
%The dynamic behavior of a system is %axiomatized by specifying each of its transition patterns 
%through a probabilistic rewrite rule of the form\\
Probabilistic behavior can be modeled with rules of the form\\
\smallskip
\quad\quad $[l] : t(\overrightarrow{x}) \longrightarrow
t'(\overrightarrow{x},\overrightarrow{y})\;\; \mbox{\textcolor{blue}{if}} \;
cond(\overrightarrow{x})\;\;\mathit{with}\;\,\mathit{probability}\;\;\overrightarrow{y} :=
\pi(\overrightarrow{x})$\\
\smallskip
\noindent where  the
 term $t'$ has new variables
$\overrightarrow{y}$ disjoint from the variables $\overrightarrow{x}$
in  $t$. 
The probabilistic nature of the rule stems from the probability
distribution $\pi(\overrightarrow{x})$,
which depends on the matching instance of $\overrightarrow{x}$, and governs the probabilistic choice of
the instance of $\overrightarrow{y}$
in the term $t'(\overrightarrow{x},\overrightarrow{y})$ according to $\pi(\overrightarrow{x})$.

\inlsec{Statistical Model Checking (SMC)}
This 
is a formal method for analyzing probabilistic systems against temporal
logic properties,
which scales  to large distributed systems~\cite{agha18,Sen05}. 
%Unlike  pure simulation,  
SMC verifies whether a property,  expressed in a stochastic temporal logic like QuaTEx~\cite{pmaude},
holds with a user-specified statistical confidence by performing Monte Carlo simulations of the system model. 
The expected value of the property 
%such as a QuaTEx expression,
is iteratively evaluated with respect to two parameters \texttt{$\alpha$} (confidence level)  and
\texttt{$\delta$} (error margin)  until  a value $\bar{v}$ is obtained such that  
with $(1-\alpha)$  statistical confidence, 
the expected value lies in the interval
$[\bar{v}-\frac{\delta}{2}, \bar{v}+\frac{\delta}{2}]$.

A Maude model is suitable for SMC analysis only if it satisfies the \emph{absence of nondeterminism} (AND) property~\cite{pmaude}.
Intuitively, this property requires that, for any given system state, the subsequent transition is uniquely determined, without any nondeterministic branching.
We establish that the model composition, as well as the model transformation (Section~\ref{sec:tool}), satisfies this guarantee.
Our framework then integrates the PVeStA tool~\cite{pvesta} for end-to-end statistical model checking of system performance properties
expressed in QuaTEx formulas.
See
Appendix~\ref{appendix:tool}
%\cite[Appendix~\ref{appendix:tool}]{tech-rpt} 
for an example QuaTEx formula for average latency.

%% file: sections/faults.tex
\section{Framework Overview}
\label{sec:overview}

%In this section, we introduce the design of our fault injection framework, which provides fine-grained fault injection capabilities for actor-based distributed system models.
%The framework currently supports seven representative types of faults: message loss, message duplication, network partition, node crash, message delay, equivocation, and message tampering.
%We first present the overall workflow, explaining how these faults are injected and propagated through the actor-based system model, and then describe the core components that implement these mechanisms.

Our formal framework composes a distributed system model with a fault injector that includes a library of fault behaviors. 
This yields an integrated model in which the system interacts with the faults through message passing.

\inlsec{Fault Library}
Seven representative faults are currently included in the fault library:
\emph{message loss}, \emph{duplication}, \emph{delay}, \emph{tampering}, \emph{equivocation}, \emph{node crash}, and \emph{network partition}. 
These cover both benign and Byzantine fault behaviors commonly
observed in realistic distributed environments and  widely exercised in deployment-level fault-injection analyses  in industrial practice~\cite{fis,twins,jepsen,ChaosMonkey}.

Each fault is modeled as an actor,  referred to as a \emph{fault handler} (see below and Section~\ref{subsec:handler}).
This design choice is motivated by our two insights. 
First,  
modeling faults as actors allows their effects to be captured through message-triggered interactions with the system, whose nodes themselves communicate via asynchronous message passing.
Second, all of these faults can be viewed as manipulations of messages exchanged between actors. 
For example, message loss, duplication, delay, tampering, and equivocation operate on individual messages: loss drops a message between two nodes; 
duplication creates a copy; 
delay injects an abnormal latency (due to network congestion or adversarial interference); 
tampering modifies the message content; 
 equivocation sends different message copies to different recipients. 
Other faults operate on \emph{sets} of messages. 
Crashing a node can be represented as dropping all of the messages destined for it; a network partition corresponds to  dropping all of the messages between two disjoint node sets, while preserving communication within each set.

This unified view not only simplifies the modeling of individual faults but also provides a common ``interface'' for extending the library with new faults.

\inlsec{Fault Injection}
Injecting faults  is carried out through the cooperation of three components within the fault injector: the 
\emph{scheduler}, the \emph{fault controller}, and the \emph{fault handlers}.
%\footnote{In the fault-free setting, the scheduler alone determines the order in which messages are delivered to nodes in the system~\cite{pmaude,DBLP:journals/pacmpl/LiuMOZB22}.}
All interactions among these components occur internally and are thus transparent to the system.
In other words, the system   continues to evolve according to its original semantics; the injector merely intercepts and manipulates in-transit messages without altering the system's transition rules.

\begin{figure}[t]
  \begin{center}
    \includegraphics[width=.9\textwidth]{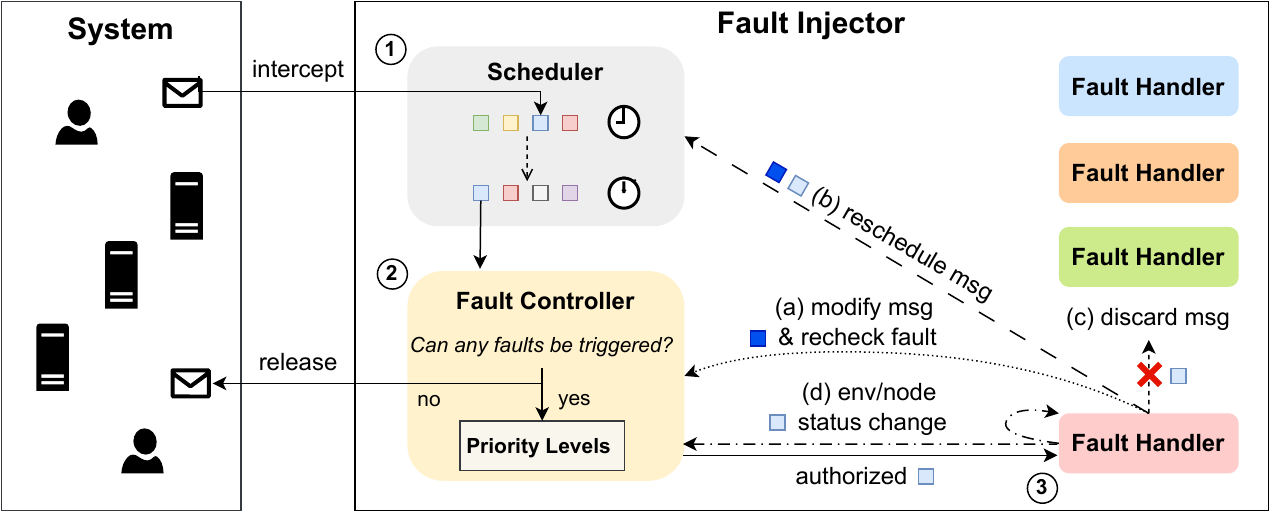}
  \end{center}
  \captionsetup{skip=0pt}
  \caption{The workflow of our fault-injection framework.
Arrows of different types in (a)--(d) represent the possible outcomes after a fault handler executes its designated operation to the message 
(Step \circled{3}).
  }
  \label{fig:overview}

\end{figure}

The overall workflow of our framework is as follows (also illustrated in Fig.~\ref{fig:overview}):

%\begin{description}
%    \item[\circled{1}]  \textbf{Message Interception and Scheduling.} 
\inlsec{\circled{1} Message Interception and Scheduling}
Fault injection begins when the scheduler intercepts an outgoing message 
\(\mathit{msg}\) 
from a sender node in the system,
e.g., the server's reply in the example of Section~\ref{sec:prelim}.
It then augments the message into the form \texttt{[\(\mathit{gt+md}\),\(\mathit{msg}\),\(\mathit{l}\)]},
     where
    % \texttt{GT + MD} 
    \(\mathit{gt+md}\)
    denotes the message's
 expected arrival time, with \(\mathit{gt}\) the current global time maintained by the scheduler and \(\mathit{md}\) a %(normal) 
 message delay sampled from a predefined  probability distribution 
 (e.g., lognormal~\cite{Benson10}),
    and 
\(\mathit{l}\) is the  label of the rewrite rule that produced the message.\footnote{More precisely, a preprocessing step annotates each outgoing message with the corresponding rule label prior to interception; see Section~\ref{sec:tool}.}
%Section~\ref{xx} details how these labels are attached  during model composition.
%The rule label is needed because many fault conditions depend on the message's origin (e.g., whether it was produced by a particular rule).

%\item[\circled{2}] \textbf{Message Scheduling.}
The scheduler maintains an ordered list of all in-transit messages, sorted by their expected arrival times.
When a message reaches the head of the list, the global clock is advanced to its expected arrival time. 
Rather than delivering the message directly to the destination node, the scheduler hands it over to the fault controller 
%as an \texttt{envelope([GT + MD, MSG, L])} 
for potential fault injection.
See
Section~\ref{subsec:scheduler} for details.

%\item[\circled{2}] 
%\textbf{Fault Control.}
\inlsec{\circled{2} Fault Control}
Upon receiving the message, the fault controller determines if any fault conditions are triggered based on  its metadata (e.g., its destination, content, and source rule label).
If a condition holds, the controller authorizes the corresponding fault behavior \(\mathit{fbhv}\) 
and passes the message to the appropriate  handler 
via 
\texttt{auth([\(\mathit{gt+md}\),\(\mathit{msg}\),\(\mathit{l}\)], \(\mathit{fbhv}\))}.
For instance, under a network partition,  \(\mathit{fbhv}\) may be \texttt{part-drop} that discards the message sent across  partitions. 

If multiple conditions hold simultaneously, the controller selects which fault behavior to authorize first according to  a predefined \emph{fault priority} (Section~\ref{subsec:controller}). 
Note that, in our framework, fault control is designed to be \emph{fine-grained}: 
a single fault  type may encompass multiple distinct fault behaviors.
For example, network partition also includes behaviors such as \texttt{part-msg}, where partitions are triggered by a specific message that marks an ``interesting'' moment in the system execution (e.g., right before a leader is selected in consensus protocols). 
%Section~\ref{subsec:controller} further elaborates on these fault controls.

Otherwise, the fault controller releases the message directly into the system configuration.
Note that the delivered message %\(\mathit{msg'}\) 
may differ from the original one if it has been modified by a fault handler, as explained below.

%\item[\circled{3}]  
%\textbf{Fault Handling.}
\inlsec{\circled{3} Fault Handling}
Once a fault behavior is authorized, the corresponding  handler applies its designated operation to the message.
This results in  one  of the following four outcomes (see Section~\ref{subsec:handler} for details):

\begin{itemize}

\item[(a)]  
The handler modifies the message and sends it
back to the fault controller 
%via \texttt{auth([GT + MD, MSG'', L])}
for
rechecking,
which may 
authorize  another fault behavior to be injected.
This arises in faults such as message tampering and equivocation.

\item[(b)] 
The handler produces one or more messages that  must be rescheduled.
For example, 
a delay fault adds an additional abnormal delay %\(\mathit{d}\) 
to the message.
%which is then inserted into the scheduler's queue.
%a duplication fault simply generates an additional message for that queue.

\item[(c)] The handler discards the message, which may result from a message-loss fault, a network partition, or delivery to a crashed node.

\item[(d)] The message  remains unchanged and is rechecked for other  faults, while the environment (e.g., a network partition or recovery) or the node's status (e.g., a crash or reboot) is altered.

\end{itemize}

\section{Formalizing Fault Injection} \label{sec:fault-injection}
In this section, we formalize the three components   within the fault injector,  namely the scheduler, the fault controller, and the fault handlers, 
together with their interactions.\footnote{
The entire Maude specification is available at
\cite{fault-maude-artifact}.}
%\url{https://github.com/ZooWagon/PerF}.

\subsection{Scheduler} \label{subsec:scheduler}
The scheduler maintains a message queue 
\texttt{msgQueue}, sorted by the messages' expected arrival times, and a global clock, \texttt{clock}, which is  advanced based on these times.
The Maude module \texttt{SCHEDULER}  below presents its declaration as an actor class, along with  its key operation  (corresponding to Step \circled{1} in Fig.~\ref{fig:overview}):\footnote{For brevity, we omit variable declarations while following the Maude convention that variables are written in capital letters.}

\begin{lstlisting}[language=maude,escapechar=~, escapeinside=``]
mod SCHEDULER is
  ...  *** omitted variable declarations and module importations
  
  class Scheduler | clock: Float, msgQueue: List{Msg} . 
  
  *** intercept and schedule an outgoing msg 
  eq [MSG, L]  < sch : Scheduler | clock: GT, msgQueue: MS > 
   = < sch : Scheduler | msgQueue: insert(MS, [GT + md, MSG, L]) > .
  *** tick global clock
  eq tick(< sch : Scheduler | clock: GT, msgQueue: [GT',MSG,L] ; MS >  OBJS) 
   = < sch : Scheduler | clock: GT', msgQueue: MS >  OBJS  [GT',MSG,L] .
endm
\end{lstlisting}

\noindent The first equation (lines 7--8) inserts an outgoing message \texttt{MSG} into the scheduler's queue with its 
 expected arrival time \texttt{GT + md}, where \texttt{md} is sampled at runtime from a certain probability distribution. 
 Note that the source rule label \texttt{L} has been captured by a
preprocessing step (Section~\ref{sec:tool}).
%, which attaches to each outgoing message the label of the rule that produced it 
The second equation (lines 10--11) advances the global clock to the corresponding  time \texttt{GT'} when the message is forwarded to the fault controller for potential manipulation.

\subsection{Fault Controller}
\label{subsec:controller}

\begin{table}[t]
\centering
\caption{Priority levels of fault behaviors. Level~1: time-triggered behaviors; Level~2: message-modification behaviors; Level~3: message-induced behaviors.}
\resizebox{\columnwidth}{!}{
\begin{tabular}{|c|c|c|c|c|} 
%\hline
\toprule
    \textbf{Behavior} &  \textbf{Description}&\textbf{Level 1 } & \textbf{Level 2 } & \textbf{Level 3} \\

\midrule
    \texttt{msg-loss} & drop an in-transit message & & $\bullet$ &  \\
%\hline

     \texttt{tampering} & modify the message content&  & $\bullet$ &  \\
%\hline
     \texttt{equivocation} & send different message copies to different receivers & & $\bullet$ &  \\
%\hline
\midrule
     \texttt{part-time} & 
      network  is partitioned when the global time reaches a specific  point
     & $\bullet$ &  &  \\

     \texttt{part-msg} & partition the network right before delivering a specific message& &  & $\bullet$ \\

     \texttt{part-drop} &drop a message across partitions&  & $\bullet$ &  \\

     \texttt{recover-time} & partition  is recovered when the global time reaches a specific  point
     &$\bullet$ &  &  \\

     \texttt{recover-msg} &recover the partition right before delivering a specific message& &  & $\bullet$ \\
%\hline
\midrule
     \texttt{crash-time} & a node crashes when the global time reaches a specific  point &$\bullet$ &  &  \\

     \texttt{crash-msg} &crash a node right before delivering a specific message & &  & $\bullet$ \\

     \texttt{crash-drop} & drop a message destined for the crashed node && $\bullet$ &  \\

     \texttt{reboot-time} & a node is rebooted when the global time reaches a specific  point& $\bullet$ &  &  \\

     \texttt{reboot-msg} & reboot a crashed node right before delivering a specific message & &  & $\bullet$ \\
%\hline
\midrule
     \texttt{duplication} & create extra copies of the message& &  & $\bullet$ \\
%\hline
   \texttt{abnormal-delay} & add an extra delay to the message & &  & $\bullet$ \\
%\hline
\bottomrule
\end{tabular}
}
\label{tb:priority}
%\vspace{-2ex}
\end{table}

Fault control is complex as our framework incorporates a wide range of faults, some of which encompass multiple behaviors.
This complexity is further amplified when 
%multiple faults interact with the system and 
several behaviors become applicable simultaneously, making it hard for the  controller to determine  when and how they should be injected effectively.

\inlsec{Prioritizing Fault Behaviors}
We introduce a set of fault-behavior priority levels
to enable fine-grained control over these behaviors.
Based on the precondition required to trigger each behavior (see below), we organize all fault behaviors into three levels,
where Level $i$ has higher priority than Level $j$ for $i < j$; therefore, behaviors at Level $i$ are always triggered before those at Level $j$ at runtime. 
If multiple fault behaviors at the same level become simultaneously applicable, the controller uniformly selects one to trigger, with the remaining behaviors considered in subsequent iterations.
Table~\ref{tb:priority} summarizes these behaviors, along with their descriptions and assigned priority levels.

\inlsecit{Level 1: Time-Triggered Behaviors}
These behaviors are triggered purely by time.\footnote{Our framework allows  users to  specify the time points at which  these behaviors are triggered; see
Appendix~\ref{appendix:tool}.
%\cite[Appendix~\ref{appendix:tool}]{tech-rpt}. %  or configure the tool to select them randomly.
} 
As the global clock may be advanced by messages whose expected arrival times go beyond the scheduled times of these behaviors, all pending ones at this level must be applied before processing any  behaviors of other levels.
%associated with those messages (i.e., Level 2 or Level 3 behaviors). 
For example, consider a network partition (\texttt{part-time}) scheduled to occur at time $t_1$, but the global time has already been advanced to $t_2$ by a 
cross-partition
message $m$, where $t_2 > t_1$. 
If the partition is not triggered first, $m$ would be delivered to its destination, which is semantically incorrect, as  no message should be  delivered across partitions once the network partition has taken effect.

\inlsecit{Level 2:  Message-Modification Behaviors}
These behaviors operate on a message while it is still mutable, i.e., modifiable or removable, assuming that all Level 1 behaviors  have already been resolved.  
They directly alter the message's content (the behaviors \texttt{tampering} 
%delivery outcome (the behavior 
and \texttt{equivocation}) or existence (the behaviors \texttt{msg-loss}, \texttt{part-drop}, and \texttt{crash-drop}).

\inlsecit{Level 3: Message-Induced Behaviors}
These behaviors require the message to be stable, i.e., its content and destination have been finalized, 
and
they act based on this finalized message without modifying it. 
Such behaviors may create new message copies (\texttt{duplication}), add an extra delay (\texttt{abnormal-delay}), or trigger message-induced changes to the system or environment  (e.g., \texttt{part-msg} for partitioning the network and \texttt{reboot-msg} for rebooting a crashed node).

\inlsec{Formalizing Fault Control}
The   module \texttt{CONTROLLER} specifies the fault controller, as shown below.
It imports all predefined fault handlers (Section~\ref{subsec:handler})  relevant to the user-specified faults (line 2).
Here, we use  message loss and network partition  to illustrate how the controller operates (referring to Step~\circled{2}).
%\todo{1. level 1; owise, GT is useless? 2. how to use L? 3. how to tell the priority? }

\begin{lstlisting}[language=maude,escapechar=~, escapeinside=``]
mod CONTROLLER is
  inc MSG-LOSS + PARTITION .  *** imported according to the user config
    
  class Controller | fbhvs : List{Bhv}, priority : Map{Bhv, Nat} .

  ceq [GT, MSG, L]  < ml : MsgLoss | >  < pt : Partition | >
      < ctrl : Controller | fbhvs: BS, priority: P >
    = < ml : MsgLoss | >  < pt : Partition | >  < ctrl : Controller | >
      (if isSatisfied(BS,[GT,MSG,L], < ml : MsgLoss | > < pt : Partition | >) 
         then auth([GT, MSG, L], B)  *** fault behavior authorized   
         else MSG fi)  *** msg released into the system config
    if B := topBhv(BS,P,[GT,MSG,L],< ml : MsgLoss | > < pt : Partition | >) .

  *** check whether message loss can be triggered
  eq isSatisfied(msg-loss, [GT, CONTENT from O to O', L], 
       < ml : MsgLoss | lossReceivers: OS, lossRules: LS, lossRate: R > HDLS) 
   = O' in OS and L in LS and rand < R .

  *** check whether network partition can be triggered by time
  eq isSatisfied(part-time, [GT, MSG, L], 
        < pt : Partition | status: S, occurTime: T > HDLS) 
   = GT >= T and S == healthy .
  ...
endm
\end{lstlisting}

The controller is modeled as an actor class with two attributes (line 4):
\texttt{fbhvs}, which specifies the list of all fault behaviors to be injected,
and \texttt{priority}, which maps each fault behavior to a natural number indicating its priority level.
Both attributes are initialized according to the user configuration; see 
Appendix~\ref{appendix:tool} 
%\cite[Appendix~\ref{appendix:tool}]{tech-rpt} 
for an example  with the message-loss and network-partition fault behaviors.

The controller's operation is defined by a conditional equation (lines 6--12), where the controller
 evaluates all fault behaviors  in \texttt{fbhvs} via the predicate \texttt{isSatisfied} to determine which behaviors are eligible for triggering (line 9).
For example, 
the \texttt{msg-loss} behavior is eligible  when 
the receiver \texttt{O'} belongs to the designated set of receivers \texttt{OS} whose incoming messages may be dropped,
the source rule label matches (by the predicate \texttt{L in LS}),  which identifies a specific execution point relevant to that receiver,
and a probabilistic draw, denoted by the random number \texttt{rand}, falls below the user-specified loss rate \texttt{R} (lines 15--17). 
The \texttt{part-time} behavior becomes eligible  when the global time exceeds the configured trigger time of the network partition (i.e., \texttt{GT >= T})
and the current status of the network is \texttt{healthy}
(lines 20--22).

When multiple  behaviors are eligible, the controller selects the highest-priority one based on  the predefined levels \texttt{P}
via  the function  \texttt{topBhv} 
(line 12). 
For instance, \texttt{part-time} is chosen over  \texttt{msg-loss} when both apply. 
The controller then forwards the message, along with the selected  behavior to the corresponding  handler (line~10), or 
releases it into the system configuration otherwise (line~11).

\subsection{Fault Handlers}
\label{subsec:handler}

A fault handler is the component  for executing the designated fault behavior once it has been authorized by the controller.
All  handlers are modeled as actors,
each with its operation defined in its own module.
Depending on the specific  behavior, a handler may produce one of the four outcomes, as described in Step~\circled{3}.

The module \texttt{PARTITION} below exemplifies how the fault handler manages three network-partition behaviors of different levels. 
For brevity, only the attributes of the \texttt{Partition} class that are relevant to these behaviors are shown.

%  ...  *** omitted var declarations, the handling of other fault bhvs, etc.
\begin{lstlisting}[language=maude]
mod PARTITION is
  class Partition | status : NetworkStatus, allNodes : Set{Oid}, 
                    parts : Partitions, occurTime : Float, ...  

  *** Level 1: a network partition is triggered by a specific time
  eq auth([GT, MSG, L], part-time)
     < pt : Partition | status: healthy, parts: [OS1 | OS2], allNodes: OS >
   = < pt : Partition | status: partitioned,  
           parts: (if OS1 == empty then randomPart(OS) else [OS1 | OS2] fi) >
     [GT, MSG, L] .

  *** Level 2: a message is dropped across partitions
  eq auth([GT,MSG,L],part-drop) < pt : Partition | > = < pt : Partition | > .
    
  *** Level 3: a network partition is triggered by a specific message
  eq auth([GT, MSG, L], part-msg)
     < pt : Partition | status: healthy, occurTime: T,
                        allNodes: OS, parts: [OS1 | OS2] >
   = < pt : Partition | status: partitioned, occurTime: GT, 
           parts: (if OS1 == empty then randomPart(OS) else [OS1 | OS2] fi) >
     [GT, MSG, L] .
  ...
endm
\end{lstlisting}

The network is partitioned either
at a specific time point (lines 6--10) or upon the arrival of a  designated message (lines 17--22).
A partition is created by switching the network status from \texttt{healthy} to \texttt{partitioned}.
The partition sets are either preconfigured, e.g., 
 \texttt{[OS1 | OS2]} denotes two disjoint node sets \texttt{OS1} and \texttt{OS2}
 that are instantiated in the initial state, 
 or generated probabilistically via the function \texttt{randomPart}.
 Once the partition is applied, the message is returned to the controller for further fault processing, as discussed earlier. 
 This corresponds to the outcome (d) in   Step~\circled{3}. 
 
During a network partition, any message sent across  the two partitions is dropped. 
This is signaled  by \texttt{part-drop}
and realized by the equation at line~13, 
corresponding to the outcome (c) in Step \circled{3}.

\subsection{Fault Injector} \label{subsec:fault-injector}
The fault injector is formed by importing 
the two modules \texttt{SCHEDULER} and \texttt{CONTROLLER}  introduced earlier,
with the latter also importing all user-specified fault handlers such as \texttt{MSG-LOSS} and  \texttt{PARTITION}.

\begin{lstlisting}[language=maude,numbers=none]
mod FAULT-INJECTOR is    inc SCHEDULER + CONTROLLER .    endm
\end{lstlisting}

As we will see in the next section, our framework automates the fault injection for quantitative performance analysis, including statistical model checking (SMC). 
To ensure that the composed system model with the fault injector is suitable for SMC, we prove that it satisfies the \emph{absence of nondeterminism} (AND) property~\cite{pmaude} (see also Section~\ref{sec:prelim}).

%Since most statistical model checkers require  the underlying system to be purely probabilistic,  we show that the composed system model with the fault injector satisfies the \emph{absence of nondeterminism} (AND) property.  Intuitively, this property requires that, for any given system state, the next  transition is uniquely determined, with no nondeterministic branching. The proof is given in Appendix~\ref{appendix:proof}.

\begin{theorem}  \label{theorem:and-comp}
The composed system model with fault injection 
guarantees AND.
\end{theorem}

Intuitively, AND holds throughout the entire reachable state space:
it holds in the initial state (base case) and is preserved by every rewrite  (inductive step), and thus cannot be violated along any execution of the composed system.
A detailed proof is given in
Appendix~\ref{appendix:proof}.

%% file: sections/impl.tex
\section{%The \ourtool Tool  
Automating Fault-Injection Analysis}
\label{sec:tool}

\begin{figure}[t]
%\begin{wrapfigure}{r}{0.5\textwidth}
  %\vspace{-6ex}
  \begin{center}
    \includegraphics[width=1\textwidth]{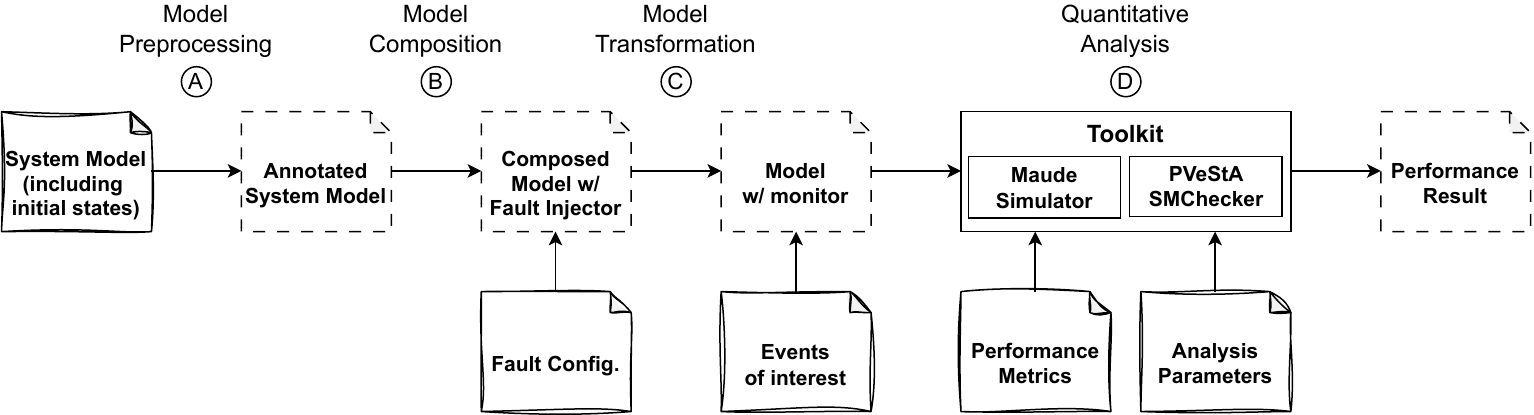}
  \end{center}
  \captionsetup{skip=0pt}
  \caption{The pipeline of \ourtool. Sketched files are provided by users, while dashed files are generated automatically by the tool. The composed model with the fault injector in Step \circled{B}  corresponds to Fig.~\ref{fig:overview}.}
  \label{fig:tool}
 % \vspace{-2ex}
%\end{wrapfigure}
\end{figure}

We have implemented our fault-injection framework in Maude and
incorporated several additional components into an automated tool, called \ourtool (available at \cite{fault-maude-artifact}), 
 to enable end-to-end performance prediction. %\footnote{The  \ourtool tool is available at \url{https://github.com/ZooWagon/PerF}.} 
The overall tool pipeline comprises four steps,
as shown in Fig.~\ref{fig:tool}:

\inlsec{\circled{A} Model Preprocessing}
The tool annotates the user-provided system model, typically nondeterministic and untimed, by attaching to each outgoing message \(\mathit{msg}\) 
%both a delay (variable) \texttt{D} and 
the label \(l\) of the rule  that generates it, yielding \texttt{[\(\mathit{msg}\),\(l\)]}. 
%At runtime, the delay is sampled from a user-specified probabilistic distribution; our framework provides a library of network latency distributions  that reflect the characteristics of real data centers~\cite{Benson10}.

\inlsec{\circled{B} Model Composition}
This step corresponds to the discussion in Section~\ref{sec:fault-injection},
where the annotated system model is composed with  the fault injector that 
consists of the scheduler, the fault controller, and the relevant fault handlers.
Before injecting any faults, the scheduler augments each message with a normal delay that is
sampled at runtime from a user-specified probability distribution.  
Our tool provides a collection of network latency distributions, including lognormal, Weibull,  and exponential,  that reflect the characteristics of real data centers~\cite{Benson10}. 
It also offers an interface
through which users can specify
both the types of faults to inject (e.g., message loss) and their injection parameters 
 (e.g., a $10\%$ loss rate for acknowledgments from followers to the leader).

\inlsec{\circled{C} Model Transformation}
The composed model is then transformed into one equipped with a monitor that  timestamps events during execution, e.g., the times at which a transaction is issued and committed.
We realize this monitoring mechanism by extending the approach introduced in~\cite{DBLP:journals/pacmpl/LiuMOZB22} to support logging multiple events of interest,  allowing users to readily define and compute performance metrics such as throughput and abort rate over these events.

\inlsec{\circled{D} Quantitative Analysis}
The toolkit performs  quantitative analysis using either Maude's built-in simulator or the  PVeStA statistical model checker~\cite{pvesta}.
In both modes,  simulations are
produced by executing the transformed model from the initial state.
For statistical model checking, 
users additionally provide performance properties expressed as QuaTEx formulas, along with experimental parameters (e.g., the statistical confidence level).
The analysis output is either the raw simulation result or the expected value of a given QuaTEx formula.

\smallskip
We illustrate these steps using the Two-Phase Commit  protocol, 
which is deferred to 
Appendix~\ref{appendix:tool}.
%\cite[Appendix~\ref{appendix:tool}]{tech-rpt}.
%due to space limitations.
%, and more examples are available at~\cite{fault-maude-artifact}.
Note  that, in addition to Theorem~\ref{theorem:and-comp}, we  prove that the model transformation in Step \circled{C} also preserves  AND, 
making the resulting model suitable for
SMC analysis.
The proof is given in 
Appendix~\ref{appendix:proof}.
%\cite[Appendix~\ref{appendix:proof}]{tech-rpt}.

\begin{comment}
Both the $F$ and $L$ transformations have been implemented in Maude.
Our framework further integrates both Maude's built-in simulator and the parallelized PVeStA statistical model checker~\cite{pvesta} into a  tool called \ourtool (available at \cite{fault-maude-artifact}), which supports end-to-end quantitative analysis of performance properties,
expressed by  QuaTEx
probabilistic temporal logic~\cite{pmaude}, over the integrated model.

Figure~\ref{fig:tool} depicts \ourtool's  pipeline. 
The F transformer 
incorporates the (extensible) fault  library \texttt{FAULT-LIB}  and 
implements the $F$ transformation.
It takes as input the system model (including the initial state) and user-specified fault parameters (e.g., which faults to inject and when), and produces an integrated model with the specified faults injected.
The L transformer realizes the $L$ transformation by additionally taking as input user-defined interesting events,
which  are logged during system executions.
Finally, the toolkit performs the quantitative analysis using the logged model, along with user-provided performance properties (e.g., throughput) defined over the logs and analysis parameters (e.g., confidence level for statistical model checking).
\todo{Appendix~\ref{appendix:tool}} demonstrates the use of \ourtool on our running example.
\end{comment}

%% file: sections/case-study.tex
\section{Case Studies}
\label{sec:exp}

We apply \ourtool to six  distributed systems of different kinds:    
(1) the widely used atomic commitment protocol 2PC,  which integrates the Cooperative Termination Protocol (CTP)~\cite{DBLP:books/aw/BernsteinHG87} to mitigate its  blocking issues during failures; %\footnote{CTP can still complete a transaction even when all partitions have prepared but only one has committed, e.g., the commit messages to all other partitions are dropped.} 
(2) the consensus protocol Raft~\cite{raft}, which achieves fault-tolerant log replication;
%for many modern coordination systems;
(3) the authoritative DNS server PowerDNS~\cite{powerdns}, which is widely deployed in production networks to ensure reliable name resolution;
(4) the quorum-based consistency protocol of Apache Cassandra~\cite{cassandra}, which balances data consistency and availability via read/write quorums;
(5) the RAMP distributed  transaction protocol~\cite{ramp}, along with its optimization, which has  been layered atop Facebook's TAO~\cite{ramp-tao}; 
and
 (6) the Byzantine fault-tolerant consensus protocol Fast-HotStuff~\cite{fast-hotstuff}, which improves the classical HotStuff design for low-latency agreement.

Table~\ref{table:exp} summarizes these case studies, including the system types, modeling and implementation efforts,  %(measured by the size of the codebase), 
and the injected faults considered in our analysis.
%We develop the formal models primarily based on the high-level or pseudo-code descriptions published by the original system designers; two exceptions are (3) and (5) whose Maude specifications are obtained from publicly available repositories. 

\begin{table}[t]

\caption{Summary of case studies.
$\dagger$:  Maude specifications obtained from the public repositories;
$\ddagger$: implementations developed in this work.}
\label{table:exp}
%{\rowcolors{1}{green!80!yellow!50}{green!70!yellow!40}
\resizebox{\textwidth}{!}{
\begin{tabular}{cccccc}
\toprule
\textbf{System} & \textbf{Kind} &\textbf{Model } & \textbf{Impl. }&\textbf{Injected Fault} & \textbf{Deployment}  \\

 & &\textbf{LoC} & \textbf{LoC}& & \textbf{Platform}  \\

\midrule

2PC with CTP & atomic commitment & 314 & 354$^\ddagger$ &  message loss&   CloudLab Utah (d6515) \\ 

Raft & consensus & 536 & 1400 & node crash &CloudLab Utah (xl170) \\

PowerDNS & name resolution & 2200$^\dagger$ & 4887 & delayed messages & DNS testbed~\cite{liu-sigcomm}\\

Cassandra Quorum &  consistency & 465 & 1402$^\ddagger$ & network partition &CloudLab Clemson 
\\ 

RAMP \& OPW & concurrency control & 2523$^\dagger$ & 7977 & message loss &Tencent Cloud \\

Fast-HotStuff & Byzantine consensus & 304 & 4247 &equiv. \& network part. &Emulab %d710 
\\

% N-Tube & botnets& bandwidth alloc. ratio&  & --\\
\bottomrule
\end{tabular}
}
%}
%\vspace{-3ex}
\end{table}

\subsection{Experimental Setup and Rationale}
\inlsec{Setup}
We predict system performance from  the formal models under 
different network latency distributions, including lognormal and exponential.
%and reuse their Maude implementations from prior work~\cite{DBLP:journals/pacmpl/LiuMOZB22}.
%Simulations are performed using Maude's built-in simulator.
We employ a cluster of Emulab machines~\cite{emulab} to parallelize the PVeStA analysis, 
%the SMC analysis with PVeStA; 
each with a 2.4 GHz quad-core Xeon processor and 12 GB of RAM.
The statistical confidence level 
$\alpha$
is set to 95\%, %($\alpha = 0.05$), 
and the error margin  $\delta$
to 0.01.
All three steps prior to SMC
are performed on a single machine,
each completing instantly.

We deploy the actual  systems across multiple platforms (see Table~\ref{table:exp}), including CloudLab (and its different clusters) and Tencent Cloud,  each  exhibiting distinct characteristics  of real-world distributed environments. 
For each system, we implement dedicated fault injectors that correspond to our model-based analysis.

%We used the lognormal or exponential
%distribution %(with parameters $\mu=0.0$ and $\sigma=1.0$)  
%for network latency 
% characterizing the network latency in realistic deployments
%\cite{Benson10,DBLP:conf/allerton/GhoshR18a}. 
%We also considered a linear modulation for message payload size and distance~\cite{DBLP:conf/networking/GuntherH05}.  

\inlsec{Evaluation Rationale}
Our evaluation rationale follows prior work~\cite{DBLP:journals/pacmpl/LiuMOZB22,liu-nfm}:
we consider a model-based prediction as accurate if it closely mirrors the observed behaviors in the deployment evaluation in terms of overall curve trends, up to a certain scaling factor.
This criterion  is motivated by  the fact that 
 model-based predictions seldom align numerically with empirical results.\footnote{Even the same implementation evaluated on different platforms
rarely yields numerically identical results due to factors like distinct network latency distributions~\cite{DBLP:journals/pacmpl/LiuMOZB22}.}
In a probabilistic model, quantitative values are inherently tied to an abstract unit of time.
This unit only approximates the real timing observed in deployed systems, and the corresponding scaling factor remains unknown at the design stage.

%%
%% exp data and plots

\pgfplotsset{height=100pt,width=200pt}
\pgfplotsset{tick style={draw=none},
legend style={draw=none}}

\begin{figure}[t!]
  % \ContinuedFloat   % 此行表示续前面的子图序号进行子图编号
   %     \centering

\begin{subfigure}{0.49\linewidth}
    \centering
    \resizebox{\textwidth}{!}{
    \begin{tikzpicture}
\begin{axis}[
    xlabel={message loss rate},
    ylabel={avg. latency (s)},
    xmin=0, xmax=0.4,
    ymin=1, ymax=5,
    xtick={0,0.05,0.1,0.15,0.2,0.25,0.3,0.35,0.4},
    xticklabels={0,0.05,0.1,0.15,0.2,0.25,0.3,0.35,0.4},
    ytick={0,1,2,3,4,5},
    legend pos=south east,
    %ymajorgrids=true,
    grid style=dashed,
]
\addplot[color=blue,mark=square,mark size=3pt] table [x=rate, y=maude, col sep=comma] {exp-fig/2pc-data-ol.csv};
\addplot[color=magenta,mark=triangle,mark size=3pt] table [x=rate, y=imp, col sep=comma] {exp-fig/2pc-data-ol.csv};
    \legend{model, deploy}
\end{axis}
\end{tikzpicture}
}
   \captionsetup{skip=0pt}
\caption{2PC with CTP (1 t.u.=10ms)}
\label{subfig:2pc}%文中引用该图片代号
\end{subfigure}
\hfill
\begin{subfigure}{0.49\linewidth}
    %\centering
    \resizebox{\textwidth}{!}{
    \begin{tikzpicture}
    \begin{axis}[
        xlabel={leader election latency (ms)},
        ylabel={cumulative \%},
        ymin=0, ymax=100,
        xmin=500, xmax=3500,
        ytick={0,25,50,75,100},
        xtick={500,1000,1500,2000,2500,3000,3500},
        legend pos=south east,
        %ymajorgrids=true,
        grid style=dashed
    ]
    \addplot[color=blue] table [x=x_mau, y=y_mau, col sep=comma] {exp-fig/raft-maude-ol.csv};
    \addplot[color=magenta,densely dotted,thick] table [x=x_imp, y=y_imp, col sep=comma] {exp-fig/raft-imp-ol.csv};
    \legend{model, deploy}
    \end{axis} 
    \end{tikzpicture} 
    }
       \captionsetup{skip=0pt}
\caption{Raft (1 t.u.=25ms)}
\label{subfig:raft}%文中引用该图片代号
\end{subfigure}

\begin{subfigure}{0.49\linewidth}
%centering
\resizebox{\textwidth}{!}{
\begin{tikzpicture}
\begin{axis}[
    xlabel={message delay (s)},
    ylabel={avg. latency (s)},
    xmin=0, xmax=1,
    ymin=0, ymax=10,
    xtick={0,0.2,0.4,0.6,0.8,1},
    ytick={0,2,4,6,8,10},
    legend pos=north west,
    %ymajorgrids=true,
    grid style=dashed
]
\addplot[color=blue,mark=square,mark size=3pt] table [x=dt, y=m-power, col sep=comma] {exp-fig/dns-data-ol.csv};
\addplot[color=magenta,mark=triangle,mark size=3pt] table [x=dt, y=imp-power, col sep=comma] {exp-fig/dns-data-ol.csv};
%\addplot[color=violet,mark=square,mark size=3pt] table [x=dt, y=m-bind, col sep=comma] {exp-fig/dns-data-ol.csv};
%\addplot[color=brown,mark=triangle,mark size=3pt] table [x=dt, y=imp-bind, col sep=comma] {exp-fig/dns-data-ol.csv};
\legend{model, deploy,Maude-Bind, Imp.-Bind}
\end{axis}
\end{tikzpicture}
}
   \captionsetup{skip=0pt}
\caption{PowerDNS (1 t.u.=1s) %\todo{remove Bind}
}
\label{subfig:dns}%文中引用该图片代号
\end{subfigure}   
\hfill
\begin{subfigure}{0.49\linewidth}
    \centering
    \resizebox{\textwidth}{!}{
    \begin{tikzpicture}
\begin{axis}[
    xlabel={partition duration (s)},
    ylabel={tput (txn/s)},
    xmin=0, xmax=25,
    ymin=8000, ymax=12000,
    xtick={0,5,10,15,20,25},
    ytick={8000,9000,10000,11000,12000},
    scaled ticks=true,
    legend style={at={(1, 1.7)},  legend columns=2, column sep=5pt},
    %legend pos=south west,
    %ymajorgrids=true,
    grid style=dashed
]
% \addplot[color=blue,mark=square,mark size=3pt]
%     coordinates {
%     (0,10568)(5,10630)(10,10340)(15,10459)(20,10352)(25,10336)
%     };
\addplot[color=blue,mark=square,mark size=3pt] table [x=t, y=mtb, col sep=comma] {temp-qc-data.csv};
% \addplot[color=magenta,mark=triangle,mark size=3pt]
%     coordinates {
%     (0,10567)(5,10404)(10,10273)(15,10502)(20,10308)(25,10300)
%     };
\addplot[color=magenta,mark=triangle,mark size=3pt] table [x=t, y=imptb, col sep=comma] {temp-qc-data.csv};
% \addplot[color=violet,mark=square,mark size=3pt]
%     coordinates {
%     (0,10139)(5,9680)(10,9423)(15,9193)(20,8841)(25,8613)
%     };
\addplot[color=orange,mark=asterisk,mark size=3pt] table [x=t, y=mta, col sep=comma] {temp-qc-data.csv};
% \addplot[color=brown,mark=triangle,mark size=3pt]
%     coordinates {
%     (0,10364)(5,10104)(10,9718)(15,9423)(20,8941)(25,8672)
%     };
\addplot[color=gray,mark=diamond,mark size=3pt] table [x=t, y=impta, col sep=comma] {temp-qc-data.csv};
% \legend{TB-Maude, TB-Imp.}
\legend{model (orig.), deploy (orig.), model (alter.), deploy (alter.)}
\end{axis}
\end{tikzpicture}
}
   \captionsetup{skip=0pt}
\caption{Cassandra Quorum (1 t.u.=1s) }
\label{subfig:qc}%文中引用该图片代号
\end{subfigure}

\begin{subfigure}{0.49\linewidth}
%\centering
\resizebox{\textwidth}{!}{
\begin{tikzpicture}
\begin{axis}[
    xlabel={message loss rate},
    ylabel={avg. latency (ms)},
    xmin=0, xmax=0.4,
    ymin=0, ymax=60,
    xtick={0,0.05,0.1,0.15,0.2,0.25,0.3,0.35,0.4},
    xticklabels={0,0.05,0.1,0.15,0.2,0.25,0.3,0.35,0.4},
    ytick={0,15,30,45,60},
    legend style={at={(1, 1.65)},  legend columns=2, column sep=5pt},
    grid style=dashed,
]
\addplot[color=blue,mark=square,mark size=3pt] table [header=true,x=lr, y=m-ctp, col sep=comma] {temp-ramp-data.csv};
\addplot[color=magenta,mark=triangle,mark size=3pt] table [header=true,x=lr, y=imp-ctp, col sep=comma] {temp-ramp-data.csv};
\addplot[color=orange,mark=asterisk,mark size=3pt] table [header=true,x=lr, y=m-opw, col sep=comma] {temp-ramp-data.csv};
\addplot[color=gray,mark=diamond,mark size=3pt] table [header=true,x=lr, y=imp-opw, col sep=comma] {temp-ramp-data.csv};
\legend{model (orig.), deploy (orig.), model (OPW), deploy (OPW)}
\end{axis}
\end{tikzpicture}
}
   \captionsetup{skip=0pt}
\caption{RAMP (1 t.u.=10ms) }
\label{subfig:ramp}
\end{subfigure}   
\hfill
\begin{subfigure}{0.49\linewidth}
%\centering
\resizebox{\textwidth}{!}{
\begin{tikzpicture}
\begin{axis}[
    xlabel={\# equivocating nodes},
    ylabel={tput (txn/s)},
    xmin=0, xmax=5,
    ymin=0, ymax=10000,
    xtick={0,1,2,3,4,5},
    ytick={0,2000,4000,6000,8000,10000},
    legend style={at={(1.1, 1.75)},  legend columns=2, column sep=0pt},
    %ymajorgrids=true,
    grid style=dashed,
]
\addplot[color=blue,mark=square,mark size=3pt] table [x=en, y=m-e, col sep=comma] {temp-fhs-data.csv};
\addplot[color=magenta,mark=triangle,mark size=3pt] table [x=en, y=imp-e, col sep=comma] {temp-fhs-data.csv};
\addplot[color=orange,mark=asterisk,mark size=3pt] table [x=en, y=m-ep, col sep=comma] {temp-fhs-data.csv};
\addplot[color=gray,mark=diamond,mark size=3pt] table [x=en, y=imp-ep, col sep=comma] {temp-fhs-data.csv};
\legend{model (eq.), deploy (eq.), model (eq. \& par.), deploy (eq. \& par.)}
\end{axis}
\end{tikzpicture}
}
   \captionsetup{skip=0pt}
\caption{Fast-HotStuff (1 t.u.=10ms)}
\label{subfig:fhs}
\end{subfigure}
   % \captionsetup{skip=5pt}
        \caption{Model-based performance predictions vs. deployment-level evaluations across six distributed systems and their variants. 
        The scaling factor 1 t.u.$=T$ ms means that one time unit in the model corresponds to $T$ ms in the deployment.
        }        
	\label{fig:exp}
%\vspace{-2ex}
\end{figure}
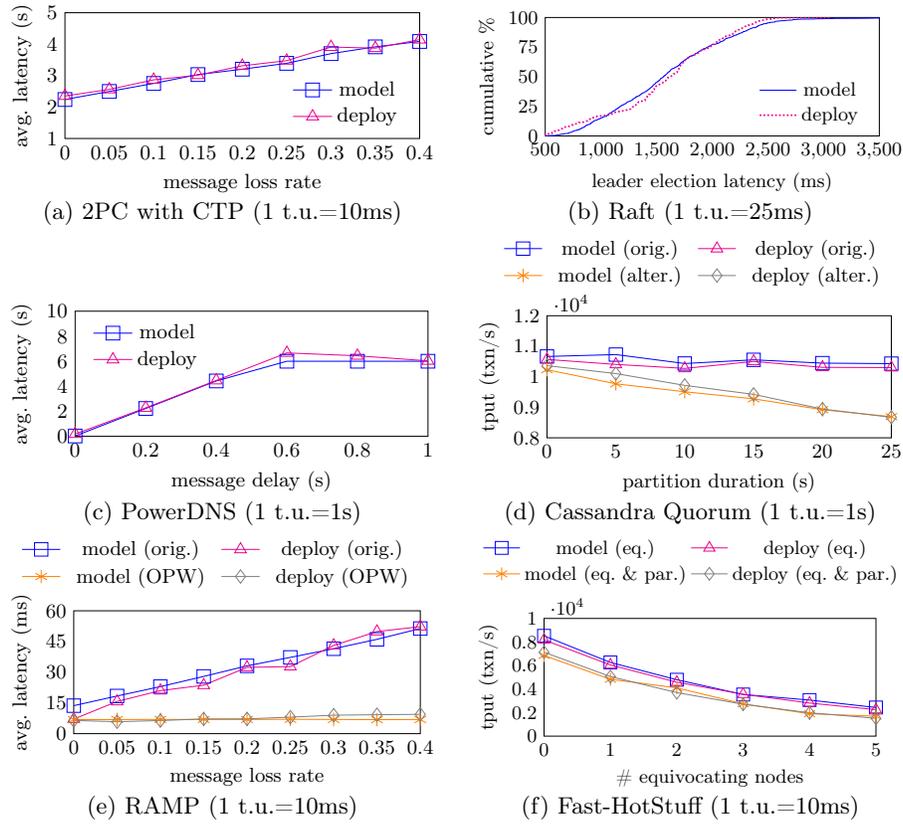

\subsection{Experimental Results}
\textbf{Summary.}
Our experimental results are shown in Figure~\ref{fig:exp}.
 Overall, across all six case studies, we demonstrate the predictive power of our framework by showing that the model-based performance  estimations, specifically throughput and latency, align closely with evaluations of the deployed systems under various types of faults. 
 These include message loss (Figure~\ref{fig:exp}(a) and (e)), node crashes (Figure~\ref{fig:exp}(b)), delayed messages (Figure~\ref{fig:exp}(c)), network partitions (Figure~\ref{fig:exp}(d)), and equivocation combined with network partitions (Figure~\ref{fig:exp}(f)).

In particular, our analysis reveals that although 
%RAMP (and likewise Cassandra's quorum protocol) 
some systems
were 
previously deemed competitive with their variants,
%~\cite{DBLP:journals/lites/LiuGRNGM17,10.1145/3639264},
substantial performance degradation can occur in either the original design or its  optimization
once faults are introduced.
This highlights the strength of \ourtool: it enables the exploration of a much broader spectrum of system behaviors  already at the design stage.
We discuss one such case, namely RAMP and its optimization, along with a mixed-fault analysis of Fast-HotStuff; 
the remaining ones are deferred to 
Appendix~\ref{appendix:exp}.
%\cite[Appendix~\ref{appendix:exp}]{tech-rpt}.

\inlsec{RAMP}
%The Read Atomic Multi-Partition (RAMP) transaction system~\cite{ramp} provides performant read and write operations for partitioned databases while ensuring the read atomicity isolation guarantee.
%In addition to the original design, the developers also introduced an alternative variant with one-phase writes (OPW). Together, these designs offer different trade-offs between system performance and data consistency.
The Read Atomic Multi-Partition (RAMP) system~\cite{ramp} 
 delivers performant transactions and the read atomicity isolation through efficient concurrency control.
In addition to the original design,
its developers introduced a faster variant with one-phase writes (OPW)   that relaxes certain consistency guarantees.

A recent study~\cite{DBLP:journals/pacmpl/LiuMOZB22} found both designs competitive, with the original design incurring only marginally higher latency. However, this conclusion holds only under fault-free conditions.
Using \ourtool, we reveal that as message loss increases, the original design exhibits significantly higher latency than the optimization OPW.
Deployment-level results corroborate our model-based predictions, as shown in Figure~\ref{fig:exp}(e).
Since message loss is inevitable in practical deployments and often intensifies under high network contention, this finding can help practitioners make more informed choices between alternative designs.

\begin{comment}
\inlsec{Cassandra Quorum}
 Cassandra~\cite{cassandra} is a distributed key–value store that employs a quorum-based mechanism to balance data consistency and availability.\footnote{A quorum specifies the number of data replicas that must acknowledge a read or write operation for it to be considered successful.}
An alternative design~\cite{bobba2018survivability} was shown to achieve performance comparable to the original mechanism, while returning more   consistent  data in certain scenarios.
However, the analysis was performed under fault-free network conditions.

As depicted in Figure~\ref{fig:exp}(d),
when network partitions occur,  the original design, by contrast, maintains more stable throughput---even when partition recovery takes longer.
We further conduct deployment-level evaluations, which are consistent with the behaviors observed in our model-based analysis.
Since network partitioning is unavoidable, especially in geo-distributed environments, this finding provides developers with deeper insight into the trade-offs between the two mechanisms in terms of system performance and data consistency.
\end{comment}

%% equivocation in blockchains
%https://ink.library.smu.edu.sg/cgi/viewcontent.cgi?params=/context/sis_research/article/8407/&path_info=3433210.3437516.pdf 

\inlsec{Fast-HotStuff}
Byzantine fault-tolerant consensus plays a pivotal role in modern blockchain systems.
The recently proposed Fast-HotStuff consensus protocol~\cite{fast-hotstuff} offers improved performance and greater robustness against forking attacks compared to the state of the art.
%Such attacks create multiple branches (forks) in the blockchain, resulting from a lack of consensus---often due to equivocation---among participants on the validity of certain blocks.

We use \ourtool to predict how Fast-HotStuff behaves under equivocation attacks~\cite{twins,10.1145/2332432.2332491}, where adversaries  propose conflicting blocks that cause ambiguity among  participants,
leading to forks and delayed consensus progress.
As shown in Figure~\ref{fig:exp}(f), increasing the number of equivocating nodes causes throughput to decline.
Moreover, 
we assess Fast-HotStuff's performance when equivocation and (benign yet prevalent) network partitions coexist.
The results indicate a moderate decrease in throughput, an expected yet reassuring outcome, as the degradation remains limited.
All our model-based predictions align closely with deployment-level evaluations conducted on a real cluster.

Notably, this mixed-fault experiment extends beyond the evaluations performed in the original study~\cite{fast-hotstuff}, providing a more comprehensive understanding of Fast-HotStuff's behavior under more realistic scenarios.
Additionally, this case study  demonstrates the modularity of our framework, which enables rapid exploration of a system's performance space under complex fault combinations.

%% file: sections/related.tex
\section{Related Work} \label{sec:discuss}

 \textbf{Formal Quantitative Analysis.}
We have already discussed in Section~\ref{sec:intro} the limitations of existing approaches to predicting distributed system performance from formal designs. They either ignore faults entirely, couple them directly with system semantics, or examine only a few isolated fault types. 
All of these works are based on Maude~\cite{maude-book} or TLA+~\cite{DBLP:books/aw/Lamport2002}, which  is largely because both formalisms provide high 
expressiveness for modeling large-scale distributed systems. 
On the other hand, 
a number of automata-based quantitative frameworks, such as \textsc{Uppaal-SMC}~\cite{uppaal-smc}, PRISM~\cite{prism},  and \textsc{Modest}~\cite{DBLP:conf/tacas/HartmannsH14},
support statistical or probabilistic model checking over timed or stochastic automata.
Nevertheless, applying these frameworks to our setting remains highly challenging, as certain system characteristics  are  difficult to capture within automata-based models,
e.g., the dynamic  creation of actors or messages that can lead to an unbounded population of such components.
Yet this is common in both normal (e.g., the joining of new nodes)
and faulty scenarios (e.g., message duplication).

\textsc{Performal}~\cite{performal} advances the formal verification of distributed system performance by providing worst-case latency bounds;
%(though it does not extend to throughput analysis);
yet such bounds are generally unattainable, e.g., in the presence of  Byzantine behaviors.
Network calculus~\cite{DBLP:journals/tit/Cruz91a} provides a mathematical foundation for reasoning about networked system performance and has been  applied to establish  latency guarantees~\cite{10.1145/2785956.2787479,DBLP:journals/tit/LiebeherrBC12}.
These approaches address an orthogonal problem;
we aim instead to capture the concrete system dynamics under faults,
covering  a wide range of  quantitative metrics beyond latency, e.g., those expressible in QuaTEx (see examples in~\cite{pmaude,DBLP:journals/pacmpl/LiuMOZB22}).

\inlsec{Posterior Fault Injection}
Recent years have witnessed a growing number of deployment-level fault injectors for assessing the reliability and resilience of distributed systems.
Academic advances include 
CoFI~\cite{cofi} for exposing network-partition failures, 
Chronos~\cite{chronos} for detecting timeout bugs, 
CrashFuzz~\cite{CrashFuzz} for uncovering crash-recovery bugs, 
%Phoenix~\cite{10.1145/3576915.3623071} for revealing blockchain resilience issues, 
and 
ByzzFuzz~\cite{Ozkan-bft} for triggering Byzantine consensus bugs.
Fault-injection techniques have also been widely adopted in production, exemplified by Jepsen~\cite{jepsen}, Netflix's Chaos Monkey~\cite{ChaosMonkey}, AWS's FIS~\cite{fis}, and Facebook's Twins~\cite{twins}.

Compared with these approaches, our framework enables early and proactive performance exploration during the design stage, rather than being reactive and applicable only after system implementation.
Moreover, our work focuses on system performance, quantifying how systems behave under diverse faults, whereas most existing approaches concentrate on  correctness and reliability issues and on how to  trigger them.
It is also worth noting that, although various benchmarks exist across different domains for evaluating  system performance,  e.g., YCSB~\cite{ycsb} for cloud database systems, they do not incorporate faults and therefore can only explore a limited portion of the overall performance space.

%% file: sections/concl.tex
\section{Conclusion}

We have presented a formal framework and tool for predicting distributed system performance  under fault conditions.
Built on Maude, our framework introduces semantics-preserving model compositions that seamlessly integrate system and fault models into a unified representation suitable for statistical analysis.
Case studies demonstrate that the framework accurately captures the performance impact of diverse and co-existing faults, aligning with empirical observations.

This work paves the way for  rapid and comprehensive exploration of the system performance space
 at an early design stage.
When paired with deployment-level fault injection, it offers a complementary means to 
better understand and improve the performance robustness of future distributed systems. 

%% file: sections/appendix.tex
\input{sections/app-tool}

\input{sections/app-proof}

\input{sections/app-exp}

%\clearpage
%\input{sections/app-handler}

%% file: sections/app-tool.tex
\section{The \ourtool Tool}  \label{appendix:tool}

%This section illustrates how \ourtool works through an example, namely the Two-Phase Commit protocol. 

\subsection{Running Example: The Two-Phase Commit Protocol}
Two-Phase Commit (2PC)  is a distributed atomic commitment protocol that ensures all participants in a transaction either commit or abort.
It proceeds in two phases. 
In the \textsc{prepare} phase, the coordinator sends a prepare request to all participants (or cohorts) and waits for their votes. 
Each cohort checks whether it can commit, e.g., by verifying local constraints, and replies with either \emph{yes} or \emph{no}. 
In the \textsc{commit} phase, if all cohorts vote \emph{yes}, the coordinator sends a commit message; otherwise, it propagates an abort. Cohorts then log the final decision. 
%2PC guarantees atomicity but can block if the coordinator fails at certain points, leaving cohorts uncertain about the final decision.
%In practice, 2PC is often  integrated with  the Cooperative Termination Protocol (CTP)~\cite{DBLP:books/aw/BernsteinHG87}  to address its inherent blocking issues (see Appendix~\ref{appendix:exp}).

Note that this is a slightly simplified version of 2PC, in which a cohort is not required to acknowledge to the coordinator that it has received the decision.
We use this simplified version here to focus on illustrating \ourtool.
In our case study, however, we consider the full integration of 2PC with the Cooperative Termination Protocol (CTP)~\cite{DBLP:books/aw/BernsteinHG87} to resolve the inherent blocking issue of 2PC.

%For each proposal, the coordinator starts a voting process among the  participants  (or cohorts) and collects their responses. 
%The proposal is approved only if all cohorts return a ``yes'' vote.  
%This protocol can be viewed as corresponding to the first phase of the well-known two-phase commit (2PC) protocol.
%The two-phase commit (2PC) protocol is the \textit{de facto} atomic commitment protocol used to ensure that a distributed transaction either commits or aborts consistently across all participating nodes. It operates in two phases. In the prepare phase, the coordinator asks all participants to prepare for commit by replying with a ``yes'' (ready) or ``no'' (abort). If all participants are ready, the protocol enters the commit phase, where the coordinator sends a commit message to finalize the transaction; otherwise, it broadcasts an abort. 
The Maude module \texttt{2PC} below specifies the protocol.
For brevity, we %show only the specification of the \textsc{prepare} phase and 
omit some auxiliary function definitions.

%\todo{the outgong msg still include md}
%  ops md pd : -> Float .  *** msg delay (md) and local processing delay (pd)
\begin{lstlisting}[language=maude]
mod 2PC is ...
  vars O O' : Oid .  var OS : Oids .  var P : Proposal .  vars V V' : Bool .
  vars PS PS' : Proposals . vars RS VS : Map{Proposal,Bool} .
  
  class Coord  | proposals : Proposals,   cohorts : Oids, 
                 waiting   : Oids,        results : Map{Proposal,Bool} .
  class Cohort | votes : Map{Proposal,Bool} .  

  op propagate_from_to_: Proposal Oid Oids -> Msgs . 
  eq propagate P from O to (O' OS) = (P from O to O') 
                                     (propagate P from O to OS) .
  eq propagate P from O to empty = null .

 crl [start] :  (start to O)  
                < O : Coord | proposals: P ; PS, waiting: empty,
                              cohorts: OS, results: RS >
            =>  < O : Coord | waiting: OS, results: insert(P,true,RS) >
                (propagate P from O to OS) if not exist(RS,P) .

  rl [vote] : (P from O' to O)  < O : Cohort | votes: (VS, P |-> V) >
            =>  < O : Cohort | >  (vote(P,V) from O to O') .

  rl [collect] : (vote(P,V') from O' to O) 
                 < O : Coord | waiting: (O' OS), results: (RS, P |-> V) >
      =>  < O : Coord | waiting: OS, results: (RS, P |-> (V and V')) > . 

  rl [decision] : < O : Coord | proposals: P ; PS, waiting: empty, 
                                cohorts: OS, results: (RS, P |-> V) >
            =>  < O : Coord | proposals: PS >  (start to O) 
                (propagate decision(P,V) from O to OS) . 

  rl [log] : (decision(P,V) from O' to O)  
             < O : Cohort | votes: (VS, P |-> V') >
            =>  < O : Cohort | votes: (VS, P |-> V) > .
endm
\end{lstlisting}

The coordinator initiates the first proposal \texttt{P} from the pool \texttt{proposals} by propagating it to the cohorts \texttt{OS} (rule \texttt{[start]}). 
Each cohort responds with its locally stored, pre-determined vote for the requested proposal, retrieved via a key–value lookup (rule \texttt{[vote]}).
Upon receiving a vote, the coordinator updates the aggregated decisions 
(through logical conjunction of the newly received and previously collected votes), removes the responder from its waiting list, and continues awaiting the remaining responses (rule \texttt{[collect]}). 
Once all votes have been collected (i.e., the waiting list becomes \texttt{empty}), the coordinator 
propagates the final decision
and
proceeds to issue the next proposal in the pool by removing the current one (rule \texttt{[decision]}).
Finally, the cohorts record the decision (rule \texttt{[log]}).
Note that, for example, the term 
\texttt{propagate\,\,P\,\,from\,\,O\,\,to\,\,OS} (line 18)
reduces to a set of messages sent from the coordinator to the cohorts,
 as defined by the  equations (lines 9--12). 
Note also that 
each voting process is triggered by a self message \texttt{start}
(lines 14 and 29).

%In our framework, the system is specified under the probabilistic message-passing actor paradigm~\cite{pmaude}.
%Hence, each voting process is triggered by a self message \texttt{start}.
%Moreover, all messages---whether between different actors or self-triggered---are associated with delays sampled from certain probabilistic distributions (see below; \todo{Section~\ref{sec:xx}} details how these delayed messages are scheduled). 

The following Maude module \texttt{INIT-2PC} specifies an initial state, 
where one coordinator and two cohorts (each initialized with pre-defined votes) participate in two proposals.
The module \texttt{2PC} described earlier is imported at  line 2.
%Message and local processing delays are modeled by lognormal and exponential distributions, respectively.

%  eq md = sampleLognormal(-3.0, 0.5) .    eq pd = sampleExponential(1000.0) .
\begin{lstlisting}[language=maude]
mod INIT-2PC is 
  inc 2PC .
  
  ops c ch1 ch2 -> Oid .  ops p1 p2 : -> Proposal .
  
  op initState : -> Config .
  eq initState = (start to c)
                 < c : Coord | proposals: p1 ; p2, cohorts: ch1 ch2, 
                               waiting: empty, results: empty >
                 < ch1 : Cohort | votes: p1 |-> true, p2 |-> false >
                 < ch2 : Cohort | votes: p1 |-> true, p2 |-> true > .
endm
\end{lstlisting}

\subsection{How \ourtool Works }
This section illustrates \ourtool's  workflow,  depicted in Fig.~\ref{fig:tool}, using 2PC as a running example.
It
consists of
four steps: 
 \emph{model preprocessing}, \emph{model composition},\emph{model transformation}, and 
 \emph{quantitative analysis}.

\inlsec{Model Preprocessing}
In this step, \ourtool annotates each outgoing message in the user-provided system model with its \emph{source rule label}.
To give users greater flexibility when modeling distributed systems, following prior work~\cite{DBLP:journals/pacmpl/LiuMOZB22}, \ourtool also equips each object-triggered rule with an \texttt{eagerEnable} equation to ensure its immediate execution once enabled.

\paragraph{Attaching Source Rule Labels.}
Whenever a rewrite rule generates messages, \ourtool rewrites that rule by wrapping the message-generation fragment with a labeling function.
This ensures that every outgoing message is annotated with the label of the rule that created it, enabling the controller to determine whether specific fault behaviors should be activated.

The following Maude specification illustrates
how the rule  \texttt{[start]}  in the module \texttt{2PC}
is preprocessed:
%the preprocessed version of the rule \texttt{[start.p]}, together with the auxiliary function \texttt{attachLabel}:

\begin{lstlisting}[language=maude,escapechar=~, escapeinside=``]
 crl [`\textcolor{violet}{start.p}`] : (start to O)    *** rule label changed  
                < O : Coord | proposals: P ; PS, waiting: empty,
                              cohorts: OS, results: RS >
            =>  < O : Coord | waiting: OS, results: insert(P,true,RS) >
                `\textcolor{violet}{attachLabel('start.p, (propagate P from O to OS))}`
                if not exist(RS,P) .

  *** auxiliary function that attaches the rule label
  op attachLabel : Qid Msgs -> LabelMsgs .
  eq attachLabel(L,MSG MSGS) = [MSG,L] attachLabel(L,MSGS) .
  eq attachLabel(L,null) = null .
\end{lstlisting}

\paragraph{Eagerness of Object-Triggered Rules.}
A user-provided system model may include rules that are not triggered by messages;
we refer to them as \emph{object-triggered rules}.
Such rules must be applied immediately once their enabling conditions become true; otherwise, nondeterministic choices could arise in firing rules. 
To enforce this semantics, the preprocessing step  generates an \texttt{eagerEnable} equation for every object-triggered rule.
This equation evaluates to \texttt{true} whenever the left-hand side of that rule matches the current configuration.

The following example illustrates this construction using the rule \texttt{[decision]}:

\begin{lstlisting}[language=maude,escapechar=~, escapeinside=``]
  rl [`\textcolor{violet}{decision.p}`] : < O : Coord | proposals: P ; PS, waiting: empty, 
                                cohorts: OS, results: (RS, P |-> V) >    
        =>  < O : Coord | proposals: PS >
        `\textcolor{violet}{attachLabel('decision.p, (start to O) (propagate decision(P,V) from O to OS))}` .

  *** eager enabling for object-triggered rule
  eq `\textcolor{violet}{eagerEnabled(< O : Coord | proposals: P ; PS, waiting: empty,}` 
                            `\textcolor{violet}{cohorts: OS, results: (RS, P |-> V) > OBJS) = true .}`
\end{lstlisting}

In addition, the scheduler's ``tick'' equation is strengthened with a guard ensuring that no object-triggered rule is eagerly enabled, as shown  below. 
This prevents the scheduler from releasing messages while an object-triggered rule is ready to fire,
which  is essential for preserving  AND   (i.e., Theorem~\ref{theorem:and-comp}).

\begin{lstlisting}[language=maude,escapechar=~, escapeinside=``]
  *** guarded tick: apply only when no object-triggered rule is eagerly enabled
 ceq tick(< sch : Scheduler | clock: GT, msgQueue: [GT',MSG,L] ; MS >  OBJS) 
   = < sch : Scheduler | clock: GT', msgQueue: MS >  OBJS  [GT',MSG,L] .
     `\textcolor{violet}{if not eagerEnable(OBJS)}` .
\end{lstlisting}

\inlsec{Model Composition}
In this step, \ourtool composes the user-provided system model with the fault injector.
Users specify which fault behaviors to inject, together with their configuration parameters such as target objects, source rule labels, probabilities, and timing,
through a dedicated interface module.

The following module illustrates how users can configure the faults \emph{message loss} and \emph{network partition}.

\begin{lstlisting}[language=maude,escapechar=~, escapeinside=``]
mod 2PC-FAULT-CONFIG is  
  inc FAULT-INJECTOR + INIT-2PC .

  *** injected fault behaviors
  eq injectedFaultBehaviors = msgloss ; part-time ; 
                              part-drop ; recover-time .

  *** configuring message loss
  eq msgLossRate = 0.3 .
  eq msgLossRules = 'decision .
  eq msgLossReceivers = ch2 .

  *** configuring network partition
  eq partitionOccurTime = 5.0 .
  eq partitionDuration = 20.0 .
  eq partitionAllNodes = (c ch1 ch2) .
  eq partitions = [c ch1 | ch2] .
endm
\end{lstlisting}
Note that  we choose to drop messages at the \texttt{[decision]} rule (line 10), 
i.e., the coordinator's decision message to a cohort.
This prevents blocking in this simplified version of 2PC and keeps the example focused on illustrating \ourtool.
%In our case study, however, we consider the full integration of 2PC with CTP to address 2PC's inherent blocking issues.

The module \texttt{INIT-2PC-FAULT} then imports this interface module, together with the module \texttt{2PC-FAULT}, which composes the system module \texttt{2PC} with the fault-injector module \texttt{FAULT-INJECTOR} (Section~\ref{subsec:fault-injector}).\footnote{The module \texttt{FAULT-INJECTOR} imports the modules \texttt{SCHEDULER} and \texttt{CONTROLLER}; the latter further imports all relevant fault handlers, such as \texttt{MSG-LOSS} and \texttt{PARTITION}.}

\begin{lstlisting}[language=maude,escapechar=~, escapeinside=``]
mod INIT-2PC-FAULT is 
  inc 2PC-FAULT-CONFIG + 2PC-FAULT .

  op initStateFault : -> Config .
  eq initStateFault = initState  *** defined in the module INIT-2PC
        < sch : Scheduler | clock: 0.0, msgQueue: nil >
        < ctrl : Controller | fbhvs: injectedFaultBehaviors,
                priority: (msgloss |-> 2, part-time |-> 1,
                           part-drop |-> 2, recover-time |-> 1) >
        < ml : MsgLoss | lossRate: msgLossRate, lossRules: msgLossRules,                          lossReceivers: msgLossReceivers >
        < pt : Partition | status: healthy, allNodes: partitionAllNodes, 
                           parts: partitions, 
                           occurTime: partitionOccurTime, 
                           duration: partitionDuration > .
endm
\end{lstlisting}

\inlsec{Model Transformation}
This step equips the composed system model with a monitor object
\texttt{< m : Monitor | events: $ES_{1}$@$T_1$ ; ... ; $ES_n$@$T_n$ >},
which records the set of user-specified events (e.g., $ES_{1}$) specified by the user together with their corresponding global timestamps (e.g., $T_1$) at runtime.
We realize this monitoring mechanism by extending prior work~\cite{DBLP:journals/pacmpl/LiuMOZB22} to support recording multiple events occurring at the same timestamp.

%To support event logging, we introduce event-related sorts and the class \texttt{Monitor}.
%Unlike the event mechanism in prior work~\cite{DBLP:journals/pacmpl/LiuMOZB22}, which allows recording only a \emph{single} event per timestamp, our design generalizes events to the sort \texttt{Events}.
%This enables multiple events to be associated with the same timestamp, reflecting richer runtime observations.
\begin{lstlisting}[language=maude]
  sorts Event Events TimedEvent TimedEvents .    
  subsort Event < Events .    subsort TimedEvent < TimedEvents .
  
  op empty : -> Events [ctor] .
  op __ : Events Events -> Events [ctor assoc comm id: empty] .
  op empty : -> TimedEvents [ctor] .
  op _;_ : TimedEvents TimedEvents -> TimedEvents [ctor assoc id: empty] .
  op _@_ : Events Float -> TimedEvent [ctor] 

  class Monitor | events: TimedEvents .

  *** user interface 
  sorts Pair EventMap .    subsort Pair < EventMap .
  
  op [_,_] : Qid Events -> Pair [ctor] .
  op none : -> EventMap [ctor] .
  op _;;_ : EventMap EventMap -> EventMap [ctor comm assoc id: none] .  
  op eventMap : -> EventMap .
\end{lstlisting}

%\paragraph{Rule transformation for event logging.}
\ourtool instruments the rules in the composed system model based on the user-specified events and their associated rule labels.
All other rules remain unchanged.
The following module illustrates how users specify the events of interest.

\begin{lstlisting}[language=maude,escapechar=~, escapeinside=``]
mod EVENT-2PC is 
  inc EVENTS + 2PC .
  
  ops propose finish : Proposal -> Event . 
  eq eventMap = ['start, propose(P)] ;; ['decision, finish(P)] [nonexec] .
endm
\end{lstlisting}

The corresponding transformed rule then logs the \texttt{propose(P)} event at time \texttt{GT}, marking the point at which \texttt{P} starts.
Note that the monitor and scheduler objects are  added automatically by \ourtool,
 and the monitor is updated on the right-hand side of the rule
 to record this timestamped event.

\begin{lstlisting}[language=maude,escapechar=~, escapeinside=``]
  crl [`\textcolor{violet}{start.p.m}`] : (start to O)
        < O : Coord | proposals: P ; PS, waiting: empty,
                      cohorts: OS, results: RS >
        `\textcolor{violet}{< m : Monitor | events: TES >  < sch : Scheduler | clock: GT >}`
     => < O : Coord | waiting: OS, results: insert(P,true,RS) >
                attachLabel('start.p.m, (propagate P from O to OS))
        `\textcolor{violet}{< m : Monitor | events: TES ; (propose(P) @ GT) >}`  `\textcolor{violet}{< sch : Scheduler | >}`
     if not exist(RS,P) . 
\end{lstlisting}

%\paragraph{Initial-state augmentation.}
Finally, the transformation  adds the monitor object 
\texttt{< m : Monitor | events: empty >}
to the initial state.

%\subsection{Preserving the AND Property} \label{app:preserve-and}
%The resulting system—obtained by applying both model composition and model transformation—preserves the operational behavior of the composed model and continues to satisfy the AND property, see Theorem~\ref{theorem:and-comp-moni}, with the proof given in Appendix~\ref{appendix:proof}.

%\begin{theorem}  \label{theorem:and-comp-moni}
%The system model composed with fault injection and equipped with the monitor guarantees AND.
%\end{theorem}

\inlsec{Quantitative Analysis}
\ourtool supports two modes of quantitative analysis:
pure simulation with Maude's built-in simulator and statistical model checking (SMC) with  PVeStA~\cite{pvesta}.
In what follows, we focus on the latter.

To evaluate performance properties such as average latency, the user provides a QuaTEx expression of the desired metric:

\begin{lstlisting}[language=maude,numbers=none]
    AvgLatency() = { s.rval(0) } ;
    eval E[ # AvgLatency() ] ; 
\end{lstlisting}
where \texttt{rval(0)} invokes the Maude function \texttt{val(0,C)} that  corresponds to the function \texttt{avglatency}.
This function returns the average proposal latency by dividing the total accumulated latency by the number of completed proposals (lines 10--11).
The relevant definitions are provided in the \texttt{ANALYSIS-2PC} module below.
Note that performance properties in general and average latency in particular are defined over the events recorded in the monitor (line 8)
% if add md in ANALYSIS-2PC, the following rows
%  *** set normal msg delay for scheduler
%  eq md = lognormal(0.0, 1.0, rand) .
\begin{lstlisting}[language=maude]
mod ANALYSIS-2PC is 
  ...

  vars T T' : Float .  var P : Proposal .  vars E E' : Events .
  vars TES TES2 TES3 : TimedEvents . 
  
  op val : Nat Config -> Float .
  eq val(0, < m : Monitor | events: TES > C) = avglatency(TES) .

  op avglatency : -> TimedEvents -> Float .   
  eq avglatency(TES) = totalLatency(TES) / numberOfProposal(TES) . 
     
  op totalLatency : TimedEvents -> Float .
  eq totalLatency(TES ; (E propose(P) @ T) ; TES2 ; 
                  (E' finish(P) @ T') ; TES3) 
   = T' - T + totalLatency(TES ; TES2 ; TES3) .
  eq totalLatency(TES) = 0.0 [owise] .
     
  op numberOfProposal : TimedEvents -> Float .
  eq numberOfProposal(TES ; (E propose(P) @ T) ; TES2 ; 
                      (E' finish(P) @ T') ; TES3) 
   = 1.0 + numberOfProposal(TES ; TES2 ; TES3) .
  eq numberOfProposal(TES) = 0.0 [owise] .
endm
\end{lstlisting}

Once these steps are completed, the user only needs to specify the standard SMC parameters (e.g., the confidence level and error margin) to carry out the  analysis, as we will see next.

%We use PVeStA to quantitatively analyze the \texttt{ANALYSIS-2PC} module with the above QuaTEx expression.
%Although earlier examples demonstrated how our framework can inject both \texttt{msgloss} and \texttt{partition} faults, the final quantitative experiment is conducted only under the message-loss fault because the 2PC running example is highly sensitive to disruptions: even mild combinations of faults (e.g., message loss together with network partition) easily prevent proposals from completing, resulting in executions where no meaningful latency can be measured. Injecting message loss alone already stresses the protocol while still allowing progress, which makes it sufficient for evaluating latency.

%Under the \texttt{msgloss} fault, running the analysis with PVeStA yields an average proposal latency of approximately $1.1$ time units (due to the probabilistic nature of the analysis, different runs may exhibit small numerical variations).

\subsection{Running the Tool}
We provide a command-line interface that automates \ourtool's entire workflow.
% The tool is invoked through the script, whose usage patterns are illustrated below:
% \begin{lstlisting}[language=bash]
% sh run.sh [--pvesta serverlist analysis-model quatex confidence threshold] protocol init fault-config [event]
% sh run.sh 2pc.maude init-2pc.maude 2pc-fault-config.maude
% sh run.sh --pvesta -l serverlist -m analysis-2pc.maude -f latency.quatex -a 0.01 -d 0.01 2pc.maude init-2pc.maude 2pc-fault-config.maude event-2pc.maude
% \end{lstlisting}

\paragraph{Command Pattern.}
\ourtool is invoked via the following general command pattern:
\begin{lstlisting}[language=bash,numbers=none]
sh run.sh [--pvesta serverlist analysis-model quatex confidence threshold] protocol init fault [event]
\end{lstlisting}

The arguments are interpreted as follows:
\begin{itemize}
    \item \texttt{protocol} and \texttt{init} are the user-provided system model and its initial-state module, e.g., \texttt{2PC} and \texttt{INIT-2PC} in our running example.
    \item \texttt{fault} specifies the fault-injection interface module, e.g., \texttt{2PC-FAULT-CONFIG}.
    \item \texttt{event} (optional) specifies the events to be monitored during the model-transformation step, e.g., \texttt{EVENT-2PC}.
\end{itemize}

The optional \texttt{pvesta} flag enables  analysis via PVeStA.
In this mode:

\begin{itemize}
\item \texttt{serverlist} lists the servers used to run simulations,
\item \texttt{analysis-model} denotes the analysis module (e.g., \texttt{ANALYSIS-2PC}),
\item \texttt{quatex} denotes the QuaTEx formula file,
\item \texttt{confidence} sets the confidence level (e.g., 0.01), and
\item \texttt{threshold} sets the error margin (e.g., 0.01).
\end{itemize}

%The optional \texttt{--pvesta} triggers quantitative analysis via PVeStA, where
%the \texttt{serverlist} contains servers for runing the simulations, the \texttt{analysis-model} denotes the analysis module (e.g.\ \texttt{ANALYSIS-2PC}), the \texttt{quatex} indicates the QuaTEx formula file, the \texttt{confidence(alpha)} is the confidence level (e.g.\ 0.01), and the \texttt{threshold(delta)} is the threshold level (e.g.\ 0.01).

\paragraph{Preprocessing and  Composition Only.}
The following invocation runs only the \emph{model preprocessing} and \emph{model composition} steps:

\begin{lstlisting}[language=bash,numbers=none]
sh run.sh 2pc.maude init-2pc.maude 2pc-fault-config.maude
\end{lstlisting}

This command produces the composed model in which the user-provided system model and the fault injector are combined with the configured faults.
No monitor is added, and no quantitative analysis is performed.

\paragraph{End-to-End Quantitative Analysis.}
The following command executes the entire workflow, including model preprocessing, model composition, model transformation, and finally  quantitative analysis with PVeStA:

\begin{lstlisting}[language=bash,numbers=none]
sh run.sh --pvesta -l serverlist -m analysis-2pc.maude -f latency.quatex -a 0.01 -d 0.01 2pc.maude init-2pc.maude 2pc-fault-config.maude event-2pc.maude
\end{lstlisting}

This invocation produces the fully instrumented model, including the monitor,  and automatically calls PVeStA to evaluate the QuaTEx formula written in \texttt{latency.quatex} using  the provided analysis parameters.

\paragraph{Analysis Result.}
Figure~\ref{fig:tool-run} shows a screenshot of \ourtool running under the full command.
The model composition and model transformation complete instantly.
Moreover, 600 simulations finish in approximately 2.27 seconds and yield an expected average latency of 1.13 (time units), using 85.578~MB of memory.

%Below we show a screenshot of our tool running under the full command, demonstrating the complete workflow including quantitative evaluation. Figure \ref{fig:tool-run} shows the analysis result returned by \ourtool. The model composition and transformation finished quickly; 600 simulations finished in about 2.27 seconds and obtained the expected average latency 1.13 with 85.578 MB memory usage.

\begin{figure}
  \centering
  \includegraphics[width=0.8\linewidth]{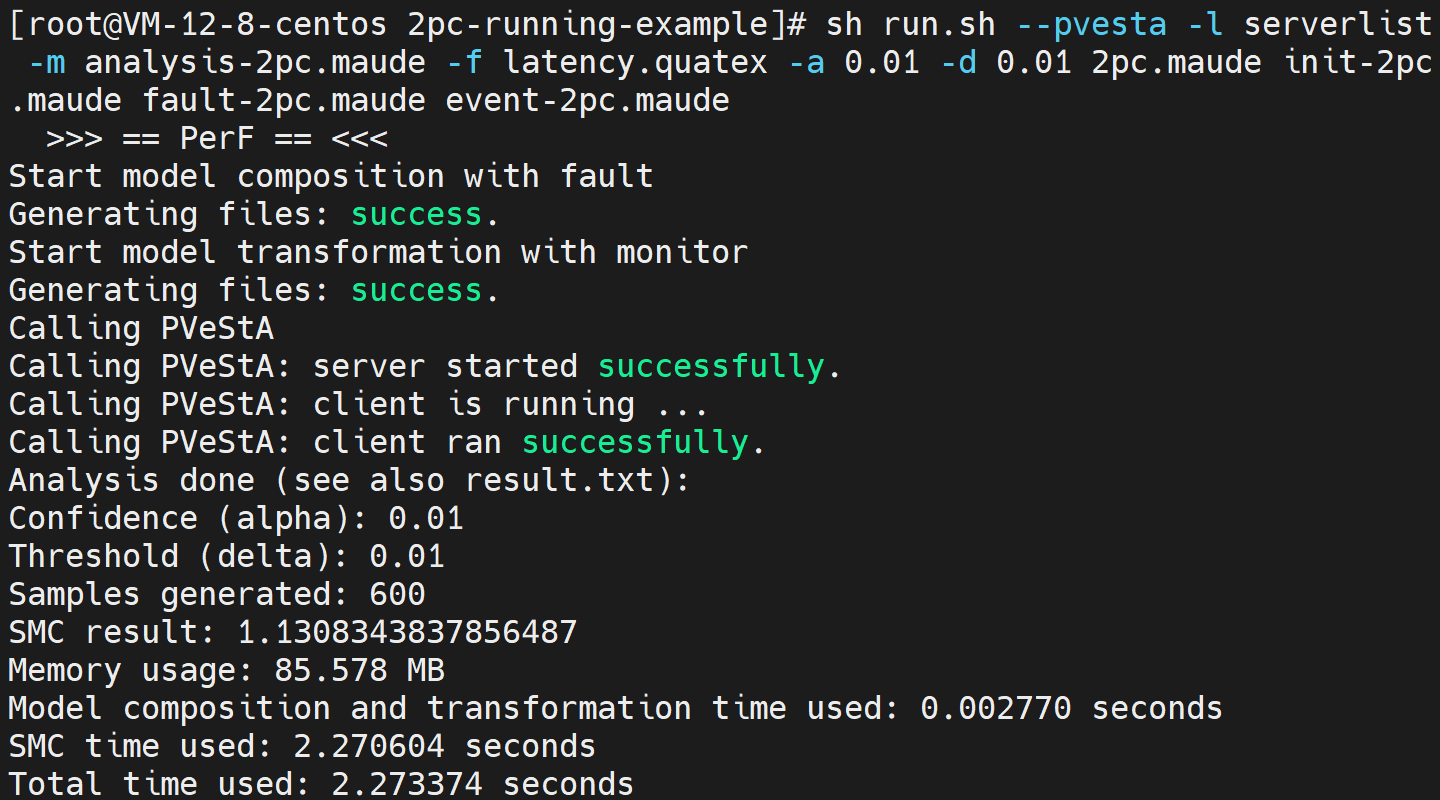}
  \captionsetup{skip=6pt}
  \caption{An example run of \ourtool with the SMC result.}
  \label{fig:tool-run}
\end{figure}

%% file: sections/app-proof.tex
\section{Proofs}
\label{appendix:proof}

This section presents proofs of the theorems that underpin \ourtool's statistical model checking analysis.
Specifically, we prove that both the model composition (with fault injection) 
and 
the model transformation (with monitoring)
satisfy 
the \emph{absence of nondeterminism} (AND) property.

\begin{definition}[AND~\cite{DBLP:journals/pacmpl/LiuMOZB22}]
With probability 1, for any reachable state
	there is at most one rewrite rule
	applicable, with a unique matching substitution. % $\theta$;
%	that can be applied to reach a next state, i.e.,
%i.e., two different rules, or the same rule but with different
%matching substitutions, can never be applied to a reachable state.
\end{definition}

\subsection{Proof of Theorem \ref{theorem:and-comp}}

Theorem \ref{theorem:and-comp} (Section~\ref{subsec:fault-injector}) establishes that the system model composed with the fault injector guarantees  the AND property.  
The proof proceeds by induction on the number of rewrite steps from the initial state.

\paragraph{Assumptions.}
The user-provided model is untimed, nondeterministic   with  two kinds of rules:
a rule that consumes a message, called a \emph{message-triggered rule},
and a rule that fires without consuming a message, called an \emph{object-triggered rule}.
Following prior work~\cite{DBLP:journals/pacmpl/LiuMOZB22},
we make the following assumptions:

\begin{enumerate}[label=(\arabic*)]
\item If a message is delivered to an object, it processes the message by applying a \emph{unique} message-triggered rule with a \emph{unique} substitution.

\item  If an object-triggered rule is enabled, the object performs its local transition by applying a \emph{unique} object-triggered rule with a \emph{unique} substitution.

\item The initial state contains exactly one  object that is enabled either by a message-triggered rule or an object-triggered rule.

%\item In any state reachable from the initial one, if a message is delivered to an object that is not enabled by any object-triggered rule, then there exists a message-triggered rule applicable to that object.

%\item Any sequence in which an object continuously applies object-triggered rules is finite.

%\item If an object is enabled by an object-triggered rule, it is enabled to perform its local transition by applying a \emph{unique} object-triggered rule with a \emph{unique} substitution.
%\item The initial configuration should only have one object, either enabled by a message-triggered rule or an object-triggered rule.
%\item In any configuration reachable from the initial one, if a message is delivered to an object which is not enabled by an object-triggered rule, there exists a message-triggered rule to apply.
%\item The steps of an object continuously applying object-triggered rules is finite.
\end{enumerate}

Assumptions (1) and (2) 
ensure local determinism: whenever an object is enabled, either by a message or by itself, there is exactly one applicable rule with a unique substitution.
Note that this does not rule out nondeterministic choice between the local transitions currently possible for an object.
Assumption~(3) imposes the same form of local determinism on the initial state by requiring that exactly one object is enabled at the start.

\begin{proof}
%We prove the claim by induction on the execution of the rewrite.

%\subsubsection{Configuration Types.}
We begin by categorizing the elements that may appear in a top-level system configuration:

\begin{itemize}
    \item \textit{\textit{qobjs}}: quiescent  objects that are not currently enabled by any \emph{object-triggered rule};
    \item \textit{\textit{eobj}}: an enabled  object for which an \emph{object-triggered rule} is  applicable;
    \item \textit{imsg}: 
    an incoming message that has been released by the fault injector and is ready to be consumed by an object (e.g.,  a term like \texttt{MSG} that appears in the configuration outside the objects);
   
    \item \textit{omsg}: 
    outgoing message(s) produced by an object and already preprocessed, waiting to be intercepted by the scheduler (a term of the form \texttt{[MSG,L]}, which is immediately captured by the scheduler's insertion equation);

    \item \textit{fmsg}: internal messages within the fault injector, such as \texttt{[GT,MSG,L]} and \texttt{auth([GT,MSG,L],B)}.
\end{itemize}

Based on these elements, we classify top-level system configurations into the following five mutually exclusive types:

\begin{enumerate}[label=(\arabic*)]
    \item \textit{qobjs eobj}: 
    an object is enabled by an object-triggered rule (i.e., an object-triggered  rule in the user-provided system model is applicable).

    \item \textit{qobjs imsg}: 
    there exists an incoming message ready for consumption.

\item \textit{qobjs}:
all objects are quiescent, with no available messages and no enabled objects.

    \item    \textit{qobjs omsg},  \textit{qobjs omsg eobj}, or \textit{qobjs omsg fmsg}:
 one or more outgoing messages exist;

    \item \textit{qobjs fmsg}: 
    an internal message within the fault injector is present  (e.g., for message inspection or authorization of fault behaviors).
    
\end{enumerate}

\noindent
Note that
Types (4) and (5) do not represent persistent rewriting states;
instead,
they are transient configurations resolved by equations, not by the rewrite rules whose nondeterministic applicability is under consideration. 
The argument below makes this distinction explicit.

We are now ready to begin the proof.
Let $n \ge 0$. For any configuration $C$ reachable from the composed initial state by $n$ rewrite steps, 
we need to show that, with probability 1, 
 $C$ has one of the forms (1)--(5) listed above, \emph{and}
at most one rewrite rule, with a unique matching substitution, is applicable to  $C$. 
We  prove by induction 
on the number of rewrite steps from the initial state.

%the number $n$ of rewrite steps from the composed initial state that every configuration reachable in $n$ steps (i) is of one of the types above and (ii) applies at most one rewrite rule with a unique matching substitution with probability 1.

\inlsec{Base Case ($n = 0$)}
The composed initial state is obtained by adding the fault-injection objects, i.e., the scheduler, the controller,  and the handlers, to the original initial state \texttt{initState}.
By Assumption~(3), \texttt{initState} contains exactly one  object, enabled either by an object-triggered rule or by a message-triggered rule.
Therefore, the composed initial state is of Type~(1) or Type~(2).
Moreover, by Assumptions~(1) and~(2), exactly one rewrite rule with a unique matching substitution is enabled in this configuration.

\inlsec{Induction Step}
Assume that the AND property holds for all configurations reachable in $n$ rewrite steps.
Let  $C$  be a configuration reachable in  $n$ steps, and consider a single rewrite step $C\longrightarrow C'$. 
By the induction hypothesis,   $C$ must have one of the five forms listed above. 
We now examine each form of  $C$ and establish that
(i) in all cases exactly one rewrite rule is applicable to $C$,
and (ii)
 the application of this rule yields a configuration $C'$ that again belongs to one of the five forms.

\paragraph{Type (1): \textit{qobjs eobj}.}
An object in the system model 
 is enabled for exactly one object-triggered rewrite rule (with a unique matching substitution).
By Assumption~(2) and by the mechanism of eagerly applying object-triggered rules (see Appendix~\ref{appendix:tool}), this rule is applied immediately, without releasing any incoming message to the actor.
The application of the rule results in one of the following outcomes, all determined by the semantics of the system model:

\begin{itemize}
    \item an outgoing message is produced (Type~(4)), and the  object may or may not remain enabled;
    \item no outgoing message is produced, but the  object \textit{eobj} remains enabled (Type (1)); or
    \item  the  object \textit{eobj} becomes quiescent and no message is produced (Type~(3)).
\end{itemize}
In all cases, exactly one rewrite rule is applied in this step, and the configuration $C'$ again belongs to one of the five types.

\paragraph{Type (2): \textit{qobjs imsg}.}
A message has been released into the configuration and is ready to be consumed by its destination object.
By Assumption~(1), the destination object has at most one enabled rewrite rule (with a unique matching substitution).
Applying this rule yields one of the following outcomes:

\begin{itemize}
\item the message is consumed and one or more outgoing messages are generated (Type~(4)), with the destination object either remaining enabled or becoming quiescent;

\item the message is consumed and the destination object becomes enabled (Type (1)); or

\item the message is consumed without producing any outgoing message, and the destination object remains quiescent after performing only internal updates (Type~(3)).

\end{itemize}

\paragraph{Type (3): \textit{qobjs}.}
No  object is enabled, and  no incoming message is pending for any object.

\begin{itemize}
    \item 
 If the scheduler holds no in-transit messages, then the system is quiescent and no rewrite rule is applicable from this point onward.
The configuration is terminal (i.e., no more rewrite steps), and the AND property trivially holds.

    \item 
If the scheduler does hold in-transit scheduled messages, the scheduler proceeds by applying the ``tick'' equation, which advances the global clock and emits an internal message \texttt{[T,MSG,L]} within the fault injector, i.e., an \textit{fmsg}.  
Since the scheduler's clock tick is realized by an equation, it does not count as a rule choice for AND.  
The configuration is transformed into Type~(5), after which the controller is invoked.  
Hence, no nondeterministic rewrite rule is introduced at this stage.

\end{itemize}

\paragraph{Type (4): \textit{qobjs omsg}, \textit{qobjs omsg eobj}, or \textit{qobjs omsg fmsg}.}
An outgoing message has been produced by a rewrite rule in the system model and appears as an \textit{omsg}.
In our framework, such outgoing messages are immediately processed by equations that attach the source rule label and insert the message into the scheduler's queue.
This is deterministic.

Consequently, a configuration of Type~(4) is transient and never serves as a rewriting target: 
it is immediately reduced by equations to a configuration of Type~(3), or to Type~(1) if an \textit{eobj} is also present, or to Type~(5) if an \textit{fmsg} is generated.
Therefore, Type~(4) does not introduce any source of  nondeterminism.

\paragraph{Type (5): \textit{qobjs fmsg}.}
This category encompasses the internal inspection and authorization messages used within the fault injector.
All transitions involving such messages are realized by equations associated with the controller and handlers. 
Therefore, configurations of Type~(5) do not introduce nondeterministic rule choices.
We now analyze the two subcases of \textit{fmsg} in detail.

\begin{itemize}
    \item \emph{\texttt{[GT,MSG,L]} is delivered to the controller.}
The controller inspects the message to determine whether any fault condition holds.
Each predicate evaluation depends only on the handler objects' attributes and on the current random sample; hence, every predicate outcome is uniquely determined.
The controller then issues at most one authorization (by selecting the behavior of the highest priority) or releases the message back to the configuration.
Both the issuance of an authorization and the release are performed by equations.

Therefore,  the controller deterministically produces one subsequent fault-injector configuration, either a Type~(5) configuration carrying an authorization or a Type~(2) configuration after a release.
     
    \item \emph{\texttt{auth([GT,MSG,L],B)} is delivered to a handler.}
    An authorization always targets a specific handler object that deterministically executes the associated behavior.
Handler behaviors fall into one of the following four categories:
    
    \begin{itemize}[label=$\bullet$]
\item \emph{Modify and recheck the message:}
the handler produces a modified internal message \texttt{[GT,MSG',L]} and returns it to the controller (yielding a Type~(5) configuration with a new \textit{fmsg});

\item \emph{Reschedule the message:}  
the handler produces an outgoing message that is inserted into the scheduler's queue (Type~(4));

\item \emph{Drop the message:}  
the handler removes the message, resulting in a Type (3) configuration;

\item \emph{Environment update:}  
the handler updates the network or node status and then returns the message to the controller (Type~(5)).

    \end{itemize}
Each authorization matches exactly one handler equation and is consumed by that equation.
Hence, the handler's operation is uniquely determined, and no  nondeterminism arises.
\end{itemize}

%Because both controller and handler steps are realized by equations that produce a deterministic outcome, type (5) configurations do not present multiple concurrent rewrite rule choices.

We have shown that each configuration $C$ has at most one applicable rewrite rule and that its successor $C'$ after one rewrite step again falls into one of the five forms.
Hence,  AND is preserved inductively, and the composed model with fault injection satisfies  this property. \qed

\end{proof}

\subsection{AND Preserved by Monitoring}

We now prove that the model transformation, which equips the composed model with a monitoring mechanism, preserves the AND property.

\begin{theorem}  \label{theorem:and-comp-moni}
The composed model equipped with the monitor guarantees AND.
\end{theorem}

\begin{proof}

Let $\mathcal{R}_{C}$ be the rewrite theory obtained after model composition.
By Theorem~\ref{theorem:and-comp}, $\mathcal{R}_{C}$ satisfies the AND property.
Let $\mathcal{R}_{C+M}$  denote the rewrite theory obtained after applying the model transformation. 
This transformation introduces two changes: (i) it adds a monitor object to the configuration, and (ii) it transforms only those rules that have user-specified events.

%\paragraph{Adding the Monitor.}
The transformation augments the initial configuration with a monitor object (see Appendix~\ref{appendix:tool}), which is 
 updated exclusively by the transformed rules corresponding to user-specified events. 
%\paragraph{Transformed Rules.}
For each rule \texttt{[l]} associated with a user-specified event, the transformation produces a corresponding  rule \texttt{[l.m]}.
This rule preserves the original rule's semantics, and its only modification is an additional update to the monitor object's attributes, which records the events. 
Rules without associated events remain unchanged.
Hence, no new rewrite rules are introduced, and the monitor does not participate in any rule-enabling conditions.

%\paragraph{Preserving Enabledness and Match Uniqueness.}
Let $C$ be any reachable configuration of $\mathcal{R}_{C+M}$, and let $h(C)$
be its projection obtained by removing the monitor object.
Then $h(C)$ is a reachable configuration of $\mathcal{R}_{C}$. 
Every rule in $\mathcal{R}_{C+M}$ is either
an unchanged rule from $\mathcal{R}_{C}$, or
a transformed rule whose applicability depends only on the non-monitor portion of the configuration.
Hence,
a rule of $\mathcal{R}_{C+M}$ is enabled in $C$ iff its underlying rule in $\mathcal{R}_{C}$ is enabled in $h(C)$.
Accordingly, the enabled-rule set of $C$ coincides with that of $h(C)$. 
As $\mathcal{R}_{C}$ satisfies AND, at most one rule is enabled in $h(C)$, and it has a unique matching substitution.
The same therefore holds in $C$: at most one (transformed or unchanged) rule is enabled, and it matches uniquely.
The additional monitor updates do not influence rule applicability or matching elsewhere in the configuration, 
thereby  introducing no nondeterminism.

We conclude that enabledness and unique matching are preserved by the model transformation, and therefore $\mathcal{R}_{C+M}$ also satisfies the AND property.  \qed
\end{proof}

%% file: sections/app-exp.tex
\section{Other Case Studies}
\label{appendix:exp}

In this section, we discuss the remaining four case studies: 
(i) Two-Phase Commit (2PC) integrated with the Cooperative Termination Protocol (CTP)~\cite{DBLP:books/aw/BernsteinHG87},
(ii) the Raft consensus protocol~\cite{raft}, 
(iii) the authoritative DNS server PowerDNS~\cite{powerdns},
and (iv) Cassandra's quorum-based consistency protocol~\cite{cassandra}.

\inlsec{2PC with CTP}
2PC is an atomic commitment protocol commonly used in distributed database systems. 
It consists of  the \textsc{prepare} phase and the \textsc{commit} phase, together ensuring the atomicity of transactions across multiple servers. 
To address its inherent blocking issues, 
CTP is often  integrated with 2PC in practical implementations.
For example, when a server $S$ 
has performed \textsc{prepare} for transaction $T$ but times out waiting for a \textsc{commit} (due to, e.g., message loss), $S$ can check $T$'s status on its sibling servers.\footnote{A server's siblings  are those servers 
associated with the same 2PC instance. For example, when a coordinator starts committing a transaction, it propagates \textsc{prepare} messages to three cohorts $S_1$, $S_2$, and $S_3$, then $S_1$ and $S_3$ are $S_2$'s siblings.}
In case a sibling $S'$ has received \textsc{commit} for $T$, then $S$ can also commit $T$. 

Figure~\ref{fig:exp}(a) shows the average latency
of 2PC with varying message loss rates. 
When combined with CTP, we observe a linear increase in the latency of 2PC as the number of undelivered \textsc{commit} messages rises.
The SMC estimation 
aligns closely with the CloudLab (Utah) evaluation. 
%with an \textit{a posteriori} scaling factor of 1 time unit $=$ 10ms.
 
\inlsec{Raft}
The past decade has seen widespread adoption of the Raft consensus algorithm across distributed systems requiring fault-tolerant consensus.
It is a leader-based algorithm in which the leader manages log replication. If the leader becomes disconnected or crashes, a new leader is elected through voting and then resolves any inconsistencies by forcing followers' logs to match its own.

We measure
the time Raft takes to detect and replace a crashed leader in a 
cluster of five servers, as shown in Figure~\ref{fig:exp}(b).
The CloudLab (Utah) deployment results closely follow the overall cumulative trend predicted by \ourtool: both exhibit a similar rise in  leader election latencies, although the deployment curve is slightly shifted to the right due to real-world network and system overheads.

\inlsec{PowerDNS}
The Domain Name System (DNS) is a hierarchical distributed system that translates human-readable domain names into IP addresses, enabling efficient Internet navigation. 
Authoritative DNS servers, such as PowerDNS (widely used in Europe), are essential components of this infrastructure, playing a critical role in name resolution.

A recently identified attack reported in~\cite{liu-sigcomm} is a slow DoS attack  in which an attacker keeps a query ``alive'' inside a resolver, consuming resources needed for legitimate traffic. 
Using \ourtool, 
we reproduce this attack in PowerDNS by automatically injecting delays 
to the responses of attacker-controlled
nameservers.
The SMC estimation is plotted  in  Figure~\ref{fig:exp}(c) alongside results from  a
 DNS testbed~\cite{liu-sigcomm}.
Overall, 
the total
query duration increases as additional delay is introduced by the malicious nameserver up to  0.6s,
and the behavior of actual resolvers closely matches our model.
In particular, 
introducing a delay of 0.6s 
can yield a cumulative end-to-end delay  of roughly 6s for resolving a single query.

\inlsec{Cassandra's Quorum-based Consistency}
 Apache Cassandra is a distributed key–value store that employs a quorum-based mechanism to balance data consistency and availability. A quorum specifies the number of data replicas that must acknowledge a read or write operation for it to be considered successful.
An alternative design~\cite{bobba2018survivability} was shown to achieve performance comparable to the original mechanism, while returning more   consistent  data in certain scenarios.
However, the analysis was performed under fault-free network conditions.

As shown in Figure~\ref{fig:exp}(d),
under network partitions the original design maintains more stable throughput, even when recovery is slower. 
The CloudLab (Clemson)
deployment evaluations further confirm the behaviors observed in our model-based analysis. 
Since network partitions are unavoidable in geo-distributed settings, these results provide developers with deeper insight into the performance–consistency trade-offs between the two mechanisms.